\documentclass[a4paper,USenglish,thm-restate,autoref,cleveref,numberwithinsect,pdfa]{llncs}

\pdfoutput=1

\title{The Modal Logic of Abstraction Refinement}

\author{Jakob Piribauer, Vinzent Zschuppe}

\institute{Technische Universit\"at Dresden, Germany}

\usepackage[T1]{fontenc}
\usepackage{graphicx}

\usepackage{amsmath,amssymb}
\usepackage{tikz}
\usetikzlibrary{positioning,automata,fit,shapes,calc,arrows.meta}
\usetikzlibrary{decorations.pathmorphing}
\tikzset{
->,
>={Stealth[length=2mm]},
node distance=2cm,
every state/.style={shape = circle,inner sep = 0pt,outer sep = 0pt,minimum size = 28pt,draw},
every loop/.style={min distance=6mm,looseness=8},
initial text=$ $
}
\usepackage{dsfont}
\usepackage{thmtools,thm-restate}
\usepackage{pgfplots}
\usepgfplotslibrary{fillbetween}
\pgfplotsset{width=10cm,compat=1.9}
\usepackage{tabularx}
    \newcolumntype{L}{>{\raggedright\arraybackslash}X}
    \newcolumntype{R}{>{\raggedleft\arraybackslash}X}
\usepackage{wrapfig}

\usepackage{thm-restate}
\usepackage{todonotes}
\usepackage{subcaption}

\usepackage{mathtools}
\usepackage{pifont} 
\usepackage{leftindex}

\usepackage{amssymb}
\usepackage{tikz}

\newcommand{\rawdiaplus}{%
  \begin{tikzpicture}
    \useasboundingbox (-0.7ex, -0.9ex) rectangle (0.7ex, 0.9ex);
    \node (w) at (-0.7ex,0) {};
    \node (e) at (+0.7ex,0) {};
    \node (s) at (0,-0.9ex) {};
    \node (n) at (0,+0.9ex) {};
    \draw (n.center) -- (e.center) -- (s.center) -- (w.center) -- (n.center);
    \draw (n.center) -- (s.center);
    \draw (e.center) -- (w.center);
  \end{tikzpicture}}

\newsavebox{\diamondplusbox}
\savebox{\diamondplusbox}{\rawdiaplus}

\newcommand{\rawdiaminus}{%
  \begin{tikzpicture}
    \useasboundingbox (-0.7ex, -0.9ex) rectangle (0.7ex, 0.9ex);
    \node (w) at (-0.7ex,0) {};
    \node (e) at (+0.7ex,0) {};
    \node (s) at (0,-0.9ex) {};
    \node (n) at (0,+0.9ex) {};
    \draw (n.center) -- (e.center) -- (s.center) -- (w.center) -- (n.center);
    \draw (e.center) -- (w.center);
  \end{tikzpicture}}

\newsavebox{\diamondminusbox}
\savebox{\diamondminusbox}{\rawdiaminus}

\newcommand{\lbox}{\mathord{\Box}}
\newcommand{\ldia}{\mathord{\Diamond}}
\newcommand{\lnext}{\ensuremath{\mathord{\mathsf{X}}}}
\newcommand{\lglobally}{\ensuremath{\mathord{\mathsf{G}}}}
\newcommand{\lfinally}{\ensuremath{\mathord{\mathsf{F}}}}
\newcommand{\luntil}{\ensuremath{\mathbin{\mathsf{U}}}}

\newcommand{\neXt}{\lnext}
\newcommand{\globally}{\lglobally}
\newcommand{\finally}{\lfinally}
\newcommand{\until}{\luntil}

\newcommand{\imp}{\rightarrow}

\newcommand{\Sat}{\Vdash}
\newcommand{\nSat}{\nVdash}
\newcommand{\sat}{\vDash}

\newcommand{\MLARfin}{\mathsf{MLAR}^{\mathit{fin}}}
\newcommand{\MLARall}{\mathsf{MLAR}^{\mathit{all}}}
\newcommand{\MLAR}{\mathsf{MLAR}}

\newcommand{\button}{\beta}

\newcommand{\Button}{\mathbf{B}}
\newcommand{\buttonSet}{\Phi}
\newcommand{\switch}{\sigma}
\newcommand{\switchB}{{\switch^\Button}}
\newcommand{\switchSet}{\Psi}

\newcommand{\weakButton}{\lambda}
\newcommand{\weakButtonAlt}{\delta}

\newcommand{\GF}{G}
\newcommand{\GFW}{\cC}
\newcommand{\GFR}{\hookrightarrow}
\newcommand{\GFw}{c}
\newcommand{\GFv}{d}
\newcommand{\GFu}{e}

\newcommand{\abstracts}{\rightsquigarrow}

\newcommand{\powerset}[1]{\mathcal{P}(#1)}

\newcommand{\cA}{\mathcal{A}}

\newcommand{\cC}{\mathcal{C}}

\newcommand{\cF}{\mathcal{F}}

\newcommand{\cL}{\mathcal{L}}

\newcommand{\cP}{\mathcal{P}}

\newcommand{\cR}{\mathcal{R}}
\newcommand{\cS}{\mathcal{S}}
\newcommand{\cT}{\mathcal{T}}

\newcommand{\AP}{\mathsf{AP}}

\newcommand{\Paths}{\mathit{Paths}}

\newcommand{\CiteAppendix}[1]{}

\usepackage{hyperref}
\usepackage{cleveref}

\begin{document}

\maketitle              

\begin{abstract}
 Iterative abstraction refinement techniques are one of the most prominent paradigms for 
 the analysis and verification of systems with large or infinite state spaces.
 This paper investigates the changes of truth values of system properties expressible in computation tree logic (CTL) when abstractions of transition systems are refined.
 To this end, the paper utilizes modal logic by defining \emph{alethic modalities} expressing possibility and necessity on top of CTL:
 The modal operator $\lozenge $  is interpreted as ``there is a refinement, in which ...'' and $\Box$ is interpreted as ``in all refinements, ...''. 
 Upper and lower bounds for the resulting \emph{modal logics of abstraction refinement}  
  are provided for three scenarios:
 1) when considering all finite abstractions of a transition system, 2) when considering all abstractions of a transition system, and 3) when considering the class of all transition systems.
Furthermore, to prove these results,  generic techniques to obtain upper bounds of modal logics using novel types of so-called control statements are developed.
\end{abstract}

\noindent\textbf{Acknowledgments.}
This work was partly funded by  the DFG Grant 389792660 as part of
TRR 248 (Foundations of Perspicuous Software Systems) and
by the BMBF (Federal Ministry of Education and Research)
in DAAD project 57616814 (SECAI, School of Embedded and
Composite AI) as part of the program Konrad Zuse Schools
of Excellence in Artificial Intelligence.

\vspace{6pt}
\noindent\textbf{Related version.}
This is the extended version of a paper accepted for publication at FoSSaCS 2026.

\section{Introduction}

Verification techniques like model checking take a mathematical model of a system as well as a formal specification expressing the intended behavior and automatically 
verify whether the specification is met by the system model. A key challenge  in practical applications is the potentially huge size of the state space of system models (see, e.g., \cite{DBLP:conf/laser/ClarkeKNZ11}). 
One of the most important approaches to tackle this problem is the use of abstractions.

\vspace{2pt}
\noindent\textbf{Abstractions.}
In a nutshell, the idea of an abstraction is to group \emph{concrete} states of a system together into \emph{abstract} states.
For every transition between concrete states, the abstraction contains a transition between the corresponding abstract states (this is called an \emph{existential abstraction} in \cite{DBLP:reference/mc/DamsG18}).
The more states are grouped together the more details of the original system are lost. A refinement of an abstraction is obtained by breaking up some of the abstract states into finer abstract states.
For an illustration, consider Figure \ref{fig:example}. The transition system $\cS$ increases or decreases $x$ by $1$ in each step. A possible abstraction $\cT_1$ and a refinement  $\cT_2$ of $\cT_1$ grouping together  states in a finer manner are depicted.

		  \begin{figure*}[t]
		  \begin{subfigure}[b]{0.8\textwidth}
\centering
    \hspace{20mm}\resizebox{.8\textwidth}{!}{%
      \begin{tikzpicture}[scale=1,auto,node distance=8mm,>=latex]
        \tikzstyle{round}=[thick,draw=black,circle]

        \node[ draw=black] (start) {$x=0$};
        \node[ left=6mm of start, draw=black] (l1) {$x=-1$};
         \node[ left=6mm of l1, draw=black] (l2) {$x=-2$};
          \node[ left=6mm of l2] (l3) {$\dots$};
          \node[ left=3mm of l3] (l4) {$\cS\colon$};
         
           \node[ right=6mm of start, draw=black] (r1) {$x=1$};
                      \node[ right=6mm of r1, draw=black] (r2) {$x=2$};

         \node[ right=6mm of r2] (r3) {$\dots$};
         
         \node[below=3mm of start] (init) {};
        
           \draw[color=black , very thick] (start) edge [bend right]  (r1);
               \draw[color=black , very thick] (start) edge [bend right]  (l1);
                       \draw[color=black , very thick] (l1) edge [bend right]  (l2);
                               \draw[color=black , very thick] (r1) edge [bend right]  (r2);
                                \draw[color=black , very thick] (r1) edge [bend right]  (start);
               \draw[color=black , very thick] (l1) edge [bend right]  (start);
                       \draw[color=black , very thick] (l2) edge [bend right]  (l1);
                               \draw[color=black , very thick] (r2) edge [bend right]  (r1);
                                  \draw[color=black , very thick] (l2) edge [bend right]  (l3);
                               \draw[color=black , very thick] (r2) edge [bend right]  (r3);
                                \draw[color=black , very thick] (l3) edge [bend right]  (l2);
                               \draw[color=black , very thick] (r3) edge [bend right]  (r2);
                               
                                \draw[color=black , very thick] (init) edge  (start);

      \end{tikzpicture}
    }
  \label{fig:conc}
  \end{subfigure}
		  \begin{subfigure}[b]{0.35\textwidth}
\centering
    \resizebox{.9\textwidth}{!}{%
      \begin{tikzpicture}[scale=1,auto,node distance=8mm,>=latex]
        \tikzstyle{round}=[thick,draw=black,circle]

        \node[draw=black] (start) {$x=0$};
        \node[left=6mm of start, draw=black] (l1) {$x<0$};
         \node[ left=3mm of l1, minimum size=12mm] (l2) {$\cT_1\colon$};
         
           \node[ right=6mm of start, draw=black] (r1) {$x>0$};
         
         \node[below=3mm of start] (init) {};
        
           \draw[color=black , very thick] (start) edge [bend right]  (r1);
               \draw[color=black , very thick] (start) edge [bend right]  (l1);
                       \draw[color=black , very thick] (l1) edge [loop below]  (l1);
                               \draw[color=black , very thick] (r1) edge [loop below]  (r1);
                                \draw[color=black , very thick] (r1) edge [bend right]  (start);
               \draw[color=black , very thick] (l1) edge [bend right]  (start);
                               
                                \draw[color=black , very thick] (init) edge  (start);

      \end{tikzpicture}
    }
  %
  \label{fig:T1}
  \end{subfigure}
  	  \begin{subfigure}[b]{0.6\textwidth}
\centering
    \resizebox{.9\textwidth}{!}{%
      \begin{tikzpicture}[scale=1,auto,node distance=8mm,>=latex]
        \tikzstyle{round}=[thick,draw=black,circle]

        \node[ draw=black] (start) {$x=0$};
        \node[ left=6mm of start, draw=black] (l1) {$x=-1$};
         \node[ left=6mm of l1, draw=black] (l2) {$x<-1$};
          \node[ left=3mm of l2,minimum size=16mm] (l3) {$\cT_2\colon$};
         
           \node[ right=6mm of start, draw=black] (r1) {$x=1$};
                      \node[ right=6mm of r1, draw=black] (r2) {$x>1$};

         
         \node[below=3mm of start] (init) {};
        
           \draw[color=black , very thick] (start) edge [bend right]  (r1);
               \draw[color=black , very thick] (start) edge [bend right]  (l1);
                       \draw[color=black , very thick] (l1) edge [bend right]  (l2);
                               \draw[color=black , very thick] (r1) edge [bend right]  (r2);
                                \draw[color=black , very thick] (r1) edge [bend right]  (start);
               \draw[color=black , very thick] (l1) edge [bend right]  (start);
                       \draw[color=black , very thick] (l2) edge [bend right]  (l1);
                               \draw[color=black , very thick] (r2) edge [bend right]  (r1);
                                \draw[color=black , very thick] (l2) edge [loop below]  (l2);
                                \draw[color=black , very thick] (r2) edge [loop below]  (r2);
                               
                                \draw[color=black , very thick] (init) edge  (start);

      \end{tikzpicture}
    }
  %
  \label{fig:T2}
  \end{subfigure}
\caption{A transition system $\cS$ and abstractions $\cT_1$ and $\cT_2$ where $\cT_2$ refines $\cT_1$.}
        \vspace{-12pt}
\label{fig:example}
\end{figure*}
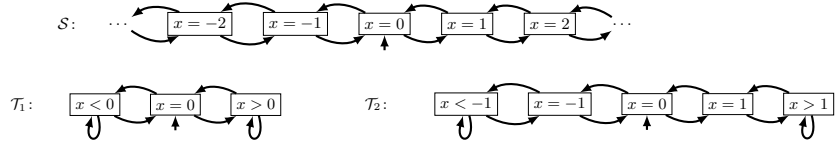

The set of executions in an abstraction  forms an over-approximation of the set of executions in the concrete system.
So, if all executions in the abstraction satisfy a  linear-time property $\varphi$, then  all executions in the underlying system satisfy $\varphi$; satisfaction of linear-time properties is \emph{preserved} under abstraction refinement.
If there are executions in an abstraction that do not satisfy $\varphi$, however, no immediate  conclusions about the underlying system are possible.

In contrast, the situation for
 branching-time logics such as \emph{computation tree logic} (CTL),  is different:
CTL uses existential and universal quantification over executions in a nested manner. Due to the existential quantification,
there is no general  preservation result for the satisfaction of CTL formulas $\Phi$ under refinement.
Consider for example the CTL formula $\Phi = \exists \neXt \exists \neXt \exists \neXt (x=0)$ stating that there is a sequence of three transitions such that afterwards the value of $x$ is $0$. This formula is true at the initial state of abstraction $\cT_1$ in Figure \ref{fig:example}. By refining the abstraction to $\cT_2$, we see that the satisfaction of $\Phi$ is not preserved.
Only for the universally quantified fragment ACTL of CTL,  satisfaction in an abstraction implies satisfaction in the underlying system in general \cite{DBLP:conf/cav/ClarkeGJLV00,DBLP:journals/jacm/ClarkeGJLV03}.

\vspace{2pt}
\noindent\textbf{Extending CTL with alethic modalities.}
To address  whether an abstraction permits conclusions about the underlying system, we propose extending CTL with \emph{alethic modalities} expressing possibility and necessity. For a CTL formula~$\Phi$, we let the modal operator $\lbox \Phi$ indicate that $\Phi$ is necessarily true, meaning it holds in all refinements of the current abstraction. Conversely, $\ldia\Phi$ expresses that $\Phi$ is possibly true, signifying that it holds in at least one refinement.
More formally, we use transition systems as system models and consider the relation  $\abstracts$, where $\cT_1\abstracts \cT_2$ means $\cT_2$ is a refinement of $\cT_1$, on the class of transition systems. 
In this way, we obtain a directed graph, called a \emph{Kripke frame}, with transition systems as \emph{worlds}.  The Kripke frame can be extended to a Kripke model by a valuation $V$ assigning to atomic propositions $p$ a set of worlds at which $p$ holds. Modal formulas are then  evaluated at a world: $\lbox \varphi$ holds at a  world if  $\varphi$ holds at all successors of the current node; $\lozenge \varphi$ holds if $\varphi$ holds at some successor.

As we aim to investigate how the satisfaction of branching-time properties evolves under abstraction refinement, 
we restrict the valuations    to range over sets of transition systems sharing a CTL-expressible property. We choose CTL as arguably the simplest prominent branching-time logic.  With this restriction, we obtain a \emph{general} Kripke frame.
The modal formulas valid on the general Kripke frame are those that hold at any world under any admissible  valuation.
These are the general principles according to which truth values of CTL-properties change under abstraction refinement.
 Informally,   propositional variables can now be understood as ``placeholders'' for CTL formulas.

\begin{example}
\label{ex:general_frame}
Consider the transition system $\cS$ depicted in \Cref{fig:general_kripke_frame}. The sets next to the states indicate the labels. We require abstractions to only group together states with the same label. An abstraction of $\cS$ hence might collapse some of the states $a_0,a_1,a_2$ as they are labelled with $a$.  If $a_0$ and $a_1$, or $a_0$ and $a_2$, are collapsed to one state, we obtain (up to isomorphism) the abstraction $\cT_1$.
Collapsing $a_1$ and $a_2$ leads to the abstraction $\cT_2$. Finally, collapsing all three states results in $\cT_3$. So, the worlds of the general Kripke frame  $\cA_\cS$ ($=\cF_{\cS}$)
consisting of all (finite) abstractions of $\cS$ are 
$\cS,\cT_1,\cT_2,\cT_3$.

Further, we observe that $\cS$ and $\cT_2$ are not distinguishable by CTL-formulas. Hence,  any valuation $V$ admissible on this general Kripke frame
satisfies the following: For each atomic proposition $p$, the set $V(p)$ contains either both $\cS$ and $\cT_2$ or neither $\cS$ nor $\cT_2$.
For example, a valuation might set $V(p)$ to the set of all transition systems that satisfy the CTL formula $\Phi=\exists \neXt b$ (i.e., whose initial states satisfy $\Phi$).
This results in $V(p) = \{\cT_1,\cT_3\}$ in the example.

\begin{figure*}[t]
\begin{center}
    \resizebox{.98\textwidth}{!}{%
        \begin{tikzpicture}[scale=1,auto,node distance=4mm,>=latex]
            \tikzstyle{rect}=[draw=white,rectangle]
            \node[initial left,label={[label distance=-3]90:{$\{a\}$}}]           (s) at ( 0, 1)  {$a_0$};
            \node[label={[label distance=-3]90:{$\{a\}$}}]                   (a1) at (2, 1.5)  {$a_1$};
            \node[label={[label distance=-3]90:{$\{a\}$}}]                   (b1) at (2, .5) {$a_2$};
            \node[label={[label distance=-3]90:{$\{b\}$}}]                   (c1) at (4, 1)      {$b$};
                        \draw
            (s)     edge[right]                                node{} (a1)
            (s)     edge[right]                                node{} (b1)
            (a1)    edge[right]                                node{} (c1)
            (b1)    edge[right]                                node{} (c1)
            (c1) edge[loop below] (c1)
                       ;  
                  
                       \draw [draw=black,fill = blue, fill opacity=0.1] (-1,2.5) rectangle (4.7,0);
                       \node (la) at (-.5,2) {$\cS$};


                        \node[initial left,label={[label distance=-3]90:{$\{a\}$}}]           (sp) at ( -7, -.5)  {$a_{02}'$};
            \node[label={[label distance=-3]90:{$\{a\}$}}]                   (a1p) at (-5, 0)  {$a_1'$};
            \node[label={[label distance=-3]90:{$\{b\}$}}]                   (c1p) at (-3, -.5)      {$b'$};
                        \draw
            (sp)     edge[right]                                node{} (a1p)
            (sp)     edge[right]                                node{} (c1p)
            (a1p)    edge[right]                                node{} (c1p)
            (c1p) edge[loop below] (c1p)
            (sp) edge[loop below] (sp)
                       ;  
                       
                       \draw [draw=black,fill = yellow, fill opacity=0.1] (-8,1) rectangle (-2.3,-1.5);
                       \node (la) at (-7.5,.5) {$\cT_1$};


                        \node[initial left,label={[label distance=-3]90:{$\{a\}$}}]           (sq) at ( 7, -.5)  {$a_0''$};
            \node[label={[label distance=-3]90:{$\{a\}$}}]                   (a1q) at (9, -.5)  {$a_{12}''$};
            \node[label={[label distance=-3]90:{$\{b\}$}}]                   (c1q) at (11, -.5)      {$b''$};
                        \draw
            (sq)     edge[right]                                node{} (a1q)
                        (a1q)    edge[right]                                node{} (c1q)
            (c1q) edge[loop below] (c1q)
                       ;  
                       
                                           \draw [draw=black,fill = blue, fill opacity=0.1] (6,1) rectangle (11.7,-1.5);
                       \node (la) at (6.5,.5) {$\cT_2$};


                                \node[initial left,label={[label distance=-3]90:{$\{a\}$}}]           (sr) at ( 1, -2)  {$a_{012}'''$};
            \node[label={[label distance=-3]90:{$\{b\}$}}]                   (c1r) at (3, -2)      {$b'''$};
                        \draw
            (sr)     edge[right]                                node{} (c1r)
                                   (c1r) edge[loop below] (c1r)
                                    (sr) edge[loop below] (sr)
                       ;     
                       
                               \draw [draw=black,fill = red, fill opacity=0.1] (-1,-.5) rectangle (4.7,-3);
                       \node (la) at (-.5,-1) {$\cT_3$};

                        \path[very thick, ->] 
                        (-1,-1.75) edge[line join=round,
decorate, decoration={
    zigzag,
    segment length=10,
    amplitude=1,post=lineto,
    post length=2pt
}] (-2.3,-1);
                                    \path[very thick, ->] 
                        (4.7,-1.75) edge[line join=round,
decorate, decoration={
    zigzag,
    segment length=10,
    amplitude=1,post=lineto,
    post length=2pt
}] (6,-1)
                        
                        ;
                                                            \path[very thick, ->] 

                        (-2.3,0) edge[line join=round,
decorate, decoration={
    zigzag,
    segment length=10,
    amplitude=1,post=lineto,
    post length=2pt
}] (-1,.7);
                                    \path[very thick, ->] 
                        (6,0) edge[line join=round,
decorate, decoration={
    zigzag,
    segment length=10,
    amplitude=1,post=lineto,
    post length=2pt
}] (4.7,.7)
                        
                        ;

        \end{tikzpicture}
    }
    \vspace{-12pt}
    \end{center}
    \caption{Example of the general frame $\cA_\cS$ ($=\cF_\cS$) resulting from transition system $\cS$ and---up to isomorphism---all its (finite) abstractions. The
    accessibility relation is  the transitive reflexive closure of the indicated refinement arrows $\abstracts$ between the worlds. 
    The transition systems $\cS$ and $\cT_2$ are not distinguishable by CTL as indicated by the filling of the rectangles.}
    \vspace{-12pt}
    \label{fig:general_kripke_frame}
\end{figure*}
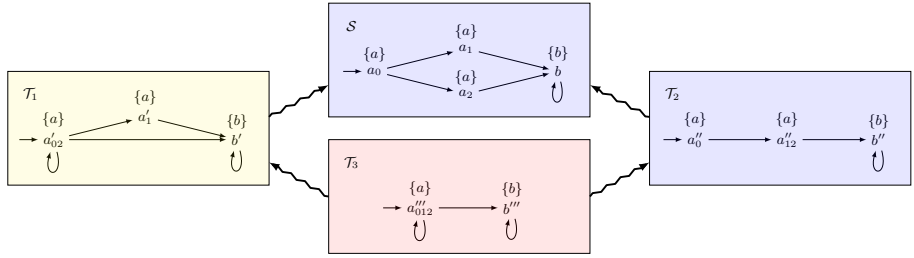

Now, e.g., the modal formula  $\Box \lozenge p \to  \lozenge \Box p$ is valid on the general frame in \Cref{fig:general_kripke_frame}:
No matter which CTL-state formula $\Phi$ corresponds to the valuation of $p$, either $\Phi$ is not true in (the initial state of) $\cS$ and so $p$ does not hold at the world $\cS$ making $\Box \lozenge p$ false everywhere, or $\cS$ satisfies $\Phi$ and so $p$ holds at the world $\cS$ and so  $\lozenge \Box p$ is true at every world of the frame.
\end{example}
\vspace{-6pt}

In this paper, we focus on determining the modal tautologies that are always true no matter which CTL-expressible properties are plugged in for the atomic propositions.
We call the resulting modal logics
\emph{modal logics of abstraction refinement (MLARs)}---which depend on the precise class of transition systems under consideration. %
 We  consider three cases:
\begin{enumerate}
\item$\MLARfin_{\cT}$ on the class $\cF_\cT$ of all finite abstractions of a transition system $\cT$,
\item $\MLARall_{\cT}$ on the class $\cA_\cT$ of all  abstractions of a transition system $\cT$, and 
\item $\MLAR$ on the class  of all transition systems.
\end{enumerate}
\begin{example}
The formula $\mathsf{(T)} = p\to \lozenge p$ belongs to all three MLARs.
The validity translates to ``For any transition system $\cT$ and any CTL-state formula $\Phi$: if $\cT\vDash \Phi$, then $\cT$ has a refinement $\cS$ with $\cS\vDash \Phi$.''
This is true  because $\cT\abstracts \cT$, i.e., $\cT$ is a refinement of itself, for any $\cT$.
The formula $p\to \Box p$ does not belong to any of the MLARs:
In \Cref{fig:example}, let the valuation of $p$ be the set of transition systems satisfying 
$\Phi = \exists \neXt \exists \neXt \exists \neXt (x=0)$. Then, $p$ holds at $\cT_1$, but $\Box p$ does not hold at $\cT_1$ as $\cT_2$ is a refinement of $\cT_1$ not satisfying $\Phi$.
\end{example}

\vspace{2pt}
\noindent\textbf{Contributions.}
We provide lower and upper bounds for the MLARs in all three cases.
The modal logics important for our main result are the well known logics  $\mathsf{S4.1}$, $\mathsf{S4.2}$ and $\mathsf{S4.2.1}$ as well as
the logic $\mathsf{S4FPF}$ denoting the modal logic of \emph{finite partial function posets}, which will be introduced in   Section \ref{sec:FPF}.
The axioms of  $\mathsf{S4.1}$, $\mathsf{S4.2}$, and $\mathsf{S4.2.1}$ are presented in
 Table \ref{tbl:axioms}.
 Our main results, also summarized in Table \ref{tbl:overview}, are as follows:
\begin{enumerate}
\item
For any $\cT$, we have $\mathsf{S4.2}\subseteq \MLARfin_{\cT}$. There is a $\cT$ with  $\mathsf{S4.2}= \MLARfin_{\cT}$.
\item
For any $\cT$, we have $\mathsf{S4.2.1}\subseteq \MLARall_{\cT}$. There is a $\cT$ with  $\mathsf{S4.2.1}= \MLARall_{\cT}$.
\item 
We have $\mathsf{S4.1}\subseteq \MLAR \subseteq \mathsf{S4.2.1}\cap \mathsf{S4FPF}$ where $\mathsf{S4FPF}$ denotes the modal logic of \emph{finite partial function posets} introduced in Section \ref{sec:FPF}.
\end{enumerate}
To prove the upper bounds, we make use of so-called \emph{control statements}, introduced in \cite{hamkins2008modal}. Control statements in our context are CTL-formulas 
that follow certain truth value patterns. E.g., a \emph{pure weak  button}\footnote{Weakness refers to the fact that the button might stay false forever, while pureness expresses that it stays true as soon as it is true for the first time.\vspace{-40pt}} is a CTL-formula that stays true in all refinements once it is true and a \emph{switch} can always be made true and be made false by further refinement.
While we  rely on the results of \cite{hamkins2008modal} to prove the upper bound $\mathsf{S4.2}$ using the existence of pure  buttons and switches, we introduce new types of control statements, namely:
\begin{enumerate}
\item \emph{restricted switches}   to prove the upper bound $\mathsf{S4.2.1}$ (\Cref{thm:S4.2.1}),
\item \emph{decisions} to prove the upper bound $\mathsf{S4FPF}$ (\Cref{thm:S4FPF}).
\end{enumerate}
 Restricted switches can be made true and false only as long as some pure button restricting the switches is not true yet. Decisions on the other hand are pairs of mutually exclusive pure weak buttons. So, by choosing refinements, one can \emph{decide} which of the pure weak buttons in the decision to make true forcing the other pure weak button to be false also in all subsequent refinements.
We show general results on how to prove these upper bounds using these new types of control statements. These results are applicable also for other classes of structures, relations between these structures, and languages describing properties of these structures and hence are of independent interest.

\begin{table}[t]
\begin{tabular}[h!]{ p{0.06\textwidth} p{0.35\textwidth} p{0.06\textwidth} p{0.38\textwidth} }
\hline
   $ \mathsf{(T)} $ & $=p \imp \ldia p$ &$ \mathsf{(4)} $ & $=\ldia \ldia p \imp \ldia p$\\
   $ \mathsf{(.2)}$ & $=\ldia \lbox p \imp \lbox \ldia p$ & $ \mathsf{(.1)}$ & $=\lbox \ldia p \imp \ldia \lbox p$ \\
   \hline
   $\mathsf{S4}$ & $=\mathsf{K+(T)+(4)}$ & $\mathsf{S4.2}$ & $=\mathsf{K+(T)+(4)+(.2)}$ \\ 
    $\mathsf{S4.1}$ & $=\mathsf{K+(T)+(4)+(.1)}$ & $\mathsf{S4.2.1}$ & $=\mathsf{K+(T)+(4)+(.2)+(.1)}$ \\
   \hline
    \end{tabular}
    \vspace{6pt}
\caption{Important  axioms and normal modal logics. For formulas $\alpha_1,\dots, \alpha_n$, we denote the smallest normal modal logic
containing $\alpha_1,\dots, \alpha_n$ by $\mathsf{K}+\alpha_1+\dots+\alpha_n$.}
\vspace{-18pt}
\label{tbl:axioms}
\end{table}

\noindent\textbf{Related work.}
Techniques on how to refine abstractions in case the current abstraction is too coarse
have been developed  for decades 
(see \cite{DBLP:reference/mc/DamsG18} for an overview).
In particular, seminal work on counter-example guided abstraction refinement (CEGAR) \cite{DBLP:conf/cav/ClarkeGJLV00,DBLP:journals/jacm/ClarkeGJLV03}, where abstractions are refined along executions $\pi$ violating a desirable property in order to potentially remove these counter-example executions from the abstraction, 
has inspired an ever-growing line of research 
(see, e.g., 
\cite{DBLP:conf/fmcad/ChauhanCKSVW02,DBLP:conf/cav/ClarkeGKS02,DBLP:conf/tacas/McMillanA03,DBLP:conf/cav/JainIGSW06}). 
Typically these approaches deal with linear-time properties of systems expressible in the universally quantified fragment ACTL of CTL.

\begin{table}[t]
\begin{tabularx}{\textwidth}{ l|l| l| L}
 class of transition systems & logic& lower bound & upper bound \\
\hline
$\cF_{\cT}$: all finite abstractions  &$\MLARfin_\cT$ & $\mathsf{S4.2}$ (Prop.~\ref{prop:lower_fin})& $\exists\cT$ with $\MLARfin_\cT = \mathsf{S4.2}$ \\[-3pt]
of  transition system $\cT$  & &&(Thm. \ref{thm:upper_fin}) 
\\
\hline
$\cA_{\cT}$: all abstractions of &   $\MLARall_\cT$ & $\mathsf{S4.2.1}$ (Prop.~\ref{prop:lower_all})&    $\exists\cT$ with $\MLARall_\cT = \mathsf{S4.2.1}$ 
\\[-3pt]
  transition system $\cT$ & & &(Thm. \ref{thm:upper_all})  \\
\hline
all transition systems &$\MLAR$ & $\mathsf{S4.1}$ (Prop.~\ref{prop:lower}) & $\subseteq \mathsf{S4.2.1}\cap \mathsf{S4FPF }$ 
(Thm.~\ref{thm:upper},~\ref{thm:upper_fpf})
\end{tabularx}
\caption{Overview of the   bounds on the modal logics of abstraction refinement.}
\label{tbl:overview}
\vspace{-18pt}
\end{table}

In work by Bozzelli et al. \cite{DBLP:journals/iandc/BozzelliDFHP14,DBLP:journals/tcs/BozzelliDP15},
the idea  to add a modality expressing truth in some/all refinements to a logic describing properties of Kripke structures is also employed. In this work, multi-agent Kripke structures modeling the epistemic state of agents are used. 
The notion of refinement, however, differs from our notion. Bozzelli et al. use the ``reverse'' of simulation as the notion of refinement meaning that a system has to simulate the systems refining it. This means that there is no restriction forbidding to remove transitions in the refinement. Further, there is---in contrast to our work---no restriction stating that states have to have successors and so there always is a ``greatest refinement'' without any transitions.
This also implies that their refinement relation is always directed, a.k.a, confluent. In contrast,  our refinement relation is not directed on the class of all transition systems. Consequently, the axiom $\mathsf{(.2)}$ is not part of $\MLAR$ in our case as shown by the fact that $\mathsf{(.2)}$ is not in $\mathsf{S4FPF}$. Only  restricted to the class of abstractions of a fixed transition system, our refinement relation is directed.

Furthermore, Bozelli et al. investigate the full logic obtained from adding the refinement modalities to a multi-modal logic, while we use CTL-definability to restrict the admissible valuations and investigate the resulting uni-modal logic only using the refinement modality. This leads to the fact that the logic Bozelli et al. is not a normal modal logic (beginning of section 5 in \cite{DBLP:journals/iandc/BozzelliDFHP14}), while we show that on the classes of finite or arbitrary abstractions of transition systems our modal logic is the normal modal logic $\mathsf{S4.2}$ or $\mathsf{S4.2.1}$, respectively.

Further, 
there are several  approaches aiming at preservation results for bran\-ching-time logics, which are orthogonal to our work.
Circumstances under which the preservation of satisfaction of  non-universal formulas can be guaranteed are investigated in
\cite{DBLP:journals/toplas/DamsGG97}.
Furthermore, 
 three-valued logics   evaluating formulas in an abstraction with truth values \emph{true}, \emph{false}, and \emph{don't know}, indicating what can be concluded about the underlying system,  have been studied (see, e.g., \cite{DBLP:conf/vmcai/GrumbergLLS05}).
In the same spirit, 
\emph{Kripke modal transition systems} (KMTS)   as abstractions  simultaneously maintain an over- and an under-approximation of the set of possible executions 
of a system by using \emph{may}- and \emph{must}-transitions (see, e.g., \cite{DBLP:conf/lics/LarsenT88,DBLP:conf/avmfss/Larsen89,DBLP:conf/esop/HuthJS01}). 

Conceptually, abstraction refinement shares similarities with
 the \emph{generalized model-checking problem} asking  whether there is a concretization  of a transition system, in which not all atomic propositions are specified, that satisfies a property \cite{Bruns2000,DBLP:conf/cav/GodefroidJ02,Godefroid2005,DBLP:journals/sttt/GodefroidP11}. 
Here, a concretization adds more information on the atomic propositions.
 There is, however, no direct relation to abstraction refinement.

The idea to use modal logic to describe the principles according to which the truth values of properties of structures can change under a model construction
was introduced in the context of set-theoretic forcing, a construction to extend models of set theory,  by Hamkins and L\"owe \cite{hamkins2008modal}. 
Subsequently, a series of work investigated the  modal logic of forcing and   of further relations between set-theoretic models  \cite{hamkins2015structural,DBLP:journals/jsyml/Yaar21,hamkins2013moving,piribauer2017modal,inamdar2016modal,hamkins2022modal}.
These ideas have been transferred to other mathematical areas such as graph theory \cite{hamkins2024modal} and group theory \cite{berger2023modal}.

\section{Preliminaries}

In the sequel, we introduce our notation. For  details on transition systems, abstractions, and temporal logics, see \cite{BK08}.
For  details on modal logic, see \cite{blackburn}.

\noindent
\textbf{Transition systems.} A \emph{transition system} is a tuple $\cT = (S,\rightarrow,I , \AP, L)$ where $S$ is a non-empty set of states, $\mathord{\rightarrow}\subseteq S\times S$ is a binary transition relation, $I \subseteq S$ is a set of initial states, 
$\AP$ is a set of atomic propositions, and $L\colon S \to 2^{\AP}$ is a  labeling function.
We require that there are no terminal states, i.e., for each $s\in S$, there is a $t\in S$ with $s\rightarrow t$.
A \emph{path} in a transition system is a sequence $s_0\,s_1\,s_2\,\dots \in S^\omega$ with  $s_i\rightarrow s_{i+1}$ for all $i \in \mathbb{N}$. We denote the set of paths starting in a state $s$ by 
$\Paths_{\cT}(s)$ or $\Paths(s)$ if $\cT$ is clear from context.
The \emph{trace} of a path $\pi = s_0\,s_1\,s_2\,\dots $ is the sequence $L(\pi) = L(s_0) \,L(s_1)\,L(s_2)\, \dots$.

\noindent
\textbf{Computation tree logic.} \emph{Computation tree logic (CTL)} is a branching-time logic whose syntax contains state and path formulas.
\emph{CTL state formulas}  over the set of atomic propositions $\AP$ are formed according to the  grammar
$
\Phi ::= \top \mid a \mid \Phi \land \Phi \mid \neg \Phi \mid \exists \varphi \mid \forall \varphi
$
where $a\in \AP$ is an atomic proposition and $\varphi$ is a \emph{CTL path formula} formed according to the  grammar
$
\varphi ::= \neXt \Phi \mid \Phi \until \Phi
$
where $\Phi$ represents CTL state formulas and $\neXt$ and $\until$ are the next-step and the until operator, respectively.
In a transition systems $\cT = (S,\rightarrow,I , \AP, L)$, the semantics of CTL state formulas over $\AP$ is given by the following recursive definition of the satisfaction relation for states $s\in S$:
\vspace{3pt}

\begin{tabular}[h!]{ p{0.12\textwidth} p{0.19\textwidth} p{0.12\textwidth} p{0.4\textwidth} }
$s\vDash \top$ &  always,   &   $s\vDash \Phi \land \Psi $& iff $s\vDash \Phi$  and  $s\vDash \Psi$,\\
$s\vDash a$ & iff $a\in L(s)$, & $s\vDash \exists \varphi $&  iff  $\pi \vDash \varphi$ for some $\pi \in \Paths(s)$,\\
$s\vDash \neg \Phi $ &  iff  not $s\vDash \Phi$, & $s\vDash \forall \varphi $&  iff  $\pi \vDash \varphi$ for all $\pi \in \Paths(s)$,
\end{tabular}
\vspace{2pt}

\noindent
with the satisfaction of path formulas $\varphi$ on paths $\pi = s_0\,s_1\,s_2\, \dots$    given by
\vspace{3pt}

\begin{tabular}[h!]{ p{0.2\textwidth} p{0.7\textwidth} }
$\pi \vDash \neXt \Phi $&  iff $s_1\vDash \Phi$, \\
$\pi \vDash \Phi \until \Psi $&  iff there is a $j\geq 0$ with $s_j \vDash \Psi$ and $s_i\vDash \Phi$ for all $i<j$.
\end{tabular}

\noindent
We use the usual Boolean abbreviations as well as the temporal abbreviations \emph{eventually} $\finally$ and \emph{always} $\globally$ defined as $\finally \Phi = \top \until \Phi$  as well as  $\exists \globally \Phi  = \neg \forall \finally \neg \Phi$ and 
$\forall \globally \Phi  = \neg \exists \finally \neg \Phi$. 
Further, we write $\forall \lnext^i \Phi$ as an abbreviation for $\forall \lnext \dots \forall \lnext \Phi$ where $\forall \lnext$ is repeated $i$ times.
We say that a transition system  $\cT$ satisfies a CTL-state formula $\Phi$ if all initial states of $\cT$ satisfy $\Phi$.

\noindent
\textbf{Abstractions and refinements.}
    Let $\mathcal{T}_i = (S_i, \rightarrow_i, I_i, \AP, L_i)$, $i=1,2$ be transition systems. We call $\mathcal{T}_1$ an \emph{abstraction} of $\mathcal{T}_2$ and $\mathcal{T}_2$ a \emph{refinement} of $\mathcal{T}_1$ if there is a surjective so-called \emph{abstraction function} $f: S_2 \rightarrow S_1$ s.t. the following conditions hold:
    \begin{enumerate}
     \item For all $s \in S_2$, $L_1(f(s))=L_2(s)$.
        \item $\rightarrow_1$ is the smallest relation satisfying that if $s \rightarrow_2 s^{\prime}$ then $f(s) \rightarrow_1 f(s^{\prime})$ for all $s,s^\prime \in S_2$.
        \item $I_1 = \{f(s) \mid s \in I_2\}.$
    \end{enumerate}
    We write  $\cT_1\abstracts \cT_2$ to denote that $\cT_2$ is a refinement of $\cT_1$. In this case, the set of traces in $\cT_1$ is a superset of the set of traces in $\cT_2$.
    The  function $f$ induces an equivalence relation on $S_2$ relating states  mapped to the same state in $S_1$.

    If $\mathcal{T}_1$ and $\mathcal{T}_2$ are  abstractions of some  transition system $\mathcal{T}_0$ with abstraction functions $f_{01}$ and $f_{02}$, we call $(\mathcal{T}_1,f_{01})$ an \emph{equivalence class preserving refinement} of $(\mathcal{T}_2,f_{02})$ if there is an abstraction function $f_{12}$ with $f_{02} = f_{12} \circ f_{01}$. We  write $(\mathcal{T}_2,f_{02})\abstracts^\equiv(\mathcal{T}_1,f_{01})$ in this case. Equivalence-class preserving  means that the  refinement and the corresponding abstraction function can be obtained by further splitting 
    up the equivalence classes in $\cT_2$ induced by $f_{02}$.

\noindent
\textbf{Modal logic.}
Formulas of \emph{modal logic} over propositional variables in a countably infinite set $\Pi$ are given by the grammar
$
\varphi ::= \top \mid p \mid \varphi \land \varphi \mid \neg \varphi \mid \lozenge \varphi
$
where $p\in \Pi$ is a propositional variable. We use the usual Boolean abbreviations as well as the abbreviation $\Box \varphi = \neg \lozenge \neg \varphi$.

A \emph{normal modal logic} is a set $\Lambda$ of modal formulas that contains all propositional tautologies as well as the K-axiom $\Box(p\to q) \to (\Box p \to \Box q)$
and that is closed under modus ponens (if $\varphi \in \Lambda$ and $\varphi \to \psi \in \Lambda$, then $\psi\in \Lambda$), uniform substitution (if $\varphi\in \Lambda$, then also
$\varphi[p/\chi]\in \Lambda$ where $\varphi[p/\chi]$ is obtained from $\varphi$ by replacing every occurrence of $p\in \Pi$ in $\varphi$ with the modal formula $\chi$), 
and generalization (if $\varphi \in \Lambda$, then $\Box \varphi \in \Lambda$).

A \emph{Kripke frame} is a pair $F=(W,R)$ where
$W$ is a non-empty set and 
$R\subseteq W\times W$ is a binary relation on $W$.
We call $W$ the \emph{domain} or \emph{universe} of the frame $F$. 
The elements of $W$ are called \emph{states} or \emph{worlds}. The relation $R$ is called the \emph{accessibility relation}. If we have $(w,v)\in R$ for $w,v\in W$, 
we also write $Rwv$ and say that $w$ \emph{can access} or \emph{sees} $v$, and that $v$ is a successor of $w$.
We call a world $r$ that can see every world of $F$ a \emph{root} of $F$.
A \emph{Kripke model} is a tuple $M=(W,R,V)$ where $(W,R)$ is a frame and $V\colon \Pi \to 2^W$ is a valuation.
We recursively define what it means for a modal formula $\varphi$ to be \emph{true} or \emph{satisfied} at a world $w\in W$, which we denote by $M,w\Vdash \varphi$:
\vspace{6pt}
\begin{tabular}[h!]{ p{0.14\textwidth} p{0.18\textwidth} p{0.18\textwidth} p{0.49\textwidth} }
$M,w\Vdash \top $&  always, &  $M,w\Vdash \psi \land \theta $&  iff  $M,w\Vdash \psi$ and $M,w\Vdash \theta$, \\
$M,w\Vdash p $& iff $ w\in V(p)$, & $M,w\Vdash \lozenge\psi $&  iff there is $v$ with $Rwv$ and $M,v\Vdash \psi$, \\
$M,w\Vdash \neg \psi $&  iff  $ M,w\not\Vdash \psi$.&&
\end{tabular}

\noindent
A modal formula $\varphi$ is \emph{valid at a state $w$} of a frame $F=(W,R)$, written $F,w\Vdash \varphi$, 
if we have $(F,V),w\Vdash \varphi$ for all valuations $V\colon \AP \to 2^W$ on $F$.
The formula $\varphi$ is \emph{valid} on $F$, written $F\Vdash \varphi$ if it is valid at all states $w\in W$.
A set of formulas $\Gamma$ is valid on a frame $F$, written $F\Vdash \Gamma$, if all $\varphi\in \Gamma$ are valid on $F$.
A (set of) formula(s) $\varphi$ is valid on a class of frames $\mathcal{C}$ if $\varphi$ is valid on all frames $F\in \mathcal{C}$.
Finally, we call the set of formulas that is valid on a class of frames $\mathcal{C}$ the \emph{logic} of $\mathcal{C}$, denoted by 
$\Lambda_{\mathcal{C}}$, which always is a normal modal logic.

\noindent
\textbf{General Kripke frames.}
A \emph{general Kripke frame} is a tuple $G=(W,R,\cA)$ where $(W,R)$ is a Kripke frame and $\cA\subseteq 2^W$ is a collection of \emph{admissible valuations}
that is closed under finite unions and complements.
A modal formula $\varphi$ is valid on $G$ at a world $w\in W$, written $G,w\Vdash \varphi$, if $M,w\Vdash \varphi$ for all Kripke models $M=(W,R,V)$ with $V(p)\in \cA$ for all $p\in \Pi$. It is valid on $G$, written $G\Vdash \varphi$ if it is valid at all $w\in W$. We denote the set of formulas valid at $w$ in $G$ by $\Lambda_{G}(w)$ and the set of formulas that are valid on $G$ by $\Lambda_G$.
Note that $\Lambda_G\supseteq \Lambda_{(W,R)}$ where $ \Lambda_{(W,R)}$ is the set of formulas valid on the frame $(W,R)$ because validity on a general frame is weaker than on a Kripke frame.

\section{Modal logic of abstraction refinement (MLAR)}

The idea behind  MLARs is to interpret $\ldia$ and $\lbox$ evaluated at a transition system $\cT$ as ``there is a refinement of $\cT$, in which ...'' and ``in all refinements of $\cT$, ...'', respectively. For atomic propositions, we allow valuations that correspond to CTL-expressible properties. For a class of transition systems $\GFW$, 
we obtain a general frame in this way: The states are the elements of $\GFW$, the accessibility relation is the refinement relation $\abstracts$ and the admissible valuations are defined by CTL-formulas.
We call the set of modal formulas valid on such a  general frame a MLAR.
In the sequel, we  define three types of MLARs on different classes of transition systems and afterwards provide  lower bounds on these logics.

\subsection{Defining the modal logic}

First, we will describe a general set-up for the modal logic of a relation $\GFR$ on a class of structures $\GFW$ with respect to a \emph{language} $\cL$ expressing properties of these structures.
By language, we mean a formalism containing formulas $\phi$ that express properties $\cP=\{\GFw\in \GFW\mid c\vDash \phi\} \subseteq \GFW$. 
We require that the properties expressible in $\cL$ are closed under complement and finite union, i.e., that $\cL$ can express negations and disjunctions.

\begin{definition}
    \label{def:G}
Let $\GFW$ be a class of structures, $\GFR$ be a reflexive  transitive binary relation on $\GFW$, $\mathcal{L}$ be a language and $\GF_{\GFW,\GFR,\cL} = (\GFW,\GFR,\mathrm{Admiss}(\mathcal{L}))$ be the general Kripke frame where
\[
\mathrm{Admiss}(\mathcal{L}) = \{\cP \subseteq \GFW \mid \text{there is a form. } \phi \text{ in } \mathcal{L} \text{ s.t. for all } \GFw \in \GFW: \GFw \sat \phi \text{ iff } \GFw \in \cP\}.
\]
The \emph{modal logic $\Lambda_{\GFW,\GFR,\cL}$ of the relation $\GFR$ w.r.t. language $\cL$ on the class  $\GFW$} is the modal logic of the general frame $\GF_{\GFW,\GFR,\cL}$, i.e., 
$
\Lambda_{{\GFW,\GFR,\cL}} = \{ \phi \mid \GF_{\GFW,\GFR,\cL} \Sat \phi \}$.
For any $\GFw \in \GFW$,  $\Lambda_{\GFW,\GFR,\cL}(\GFw)$ is the set of all valid formulas on $\GF_{\GFW,\GFR,\cL}$ at $\GFw$, i.e.,
$\Lambda_{{\GFW,\GFR,\cL}}(\GFw) = \{ \phi \mid \GF_{\GFW,\GFR,\cL},\GFw \Sat \phi \}$.
\end{definition}

\begin{example}
For an example instantiation of this definition consider the modal logic of Abelian groups studied in \cite{berger2023modal}:
The class $\GFW$ is the class of all Abelian groups. The relation $A \GFR B$ expresses that $A$ is  a subgroup of $B$. The language $\mathcal{L}$
is first-order logic. 
A formula in the resulting modal logic is a formula that holds for any Abelian group when $\lozenge$ is interpreted as ``there is a super-group in which'' and atomic propositions are replaced by any first-order sentences in the language of groups.
 In \cite{berger2023modal}, it is shown that the resulting logic is $\mathsf{S4.2}$.
 \end{example}

Now, we  can formally define our   \emph{modal logics of abstraction refinement}:
\begin{enumerate}
\item For a transition system $\cT$, the logic $\MLARfin_{\cT}$ is the modal logic of the refinement relation $\abstracts$ w.r.t. the language CTL on the class $\cF_\cT$ of all finite abstractions of $\cT$, i.e., 
$\MLARfin_{\cT} = \Lambda_{\cF_\cT, \abstracts, \mathrm{CTL}}$.
\item For a transition system $\cT$, $\MLARall_{\cT}$ is the modal logic of $\abstracts$ w.r.t.  CTL on the class $\cA_\cT$ of all  abstractions of $\cT$, i.e., 
$\MLARall_{\cT}= \Lambda_{\cA_\cT, \abstracts, \mathrm{CTL}}$
\item 
The logic $\MLAR$ is the modal logic of $\abstracts$ w.r.t.  CTL on the class $\mathfrak{A}$ of all transition systems, i.e., 
$\MLAR= \Lambda_{\mathfrak{A}, \abstracts, \mathrm{CTL}}$
\end{enumerate}

Note that the choice of the language determines the admissible valuations and hence influences the MLARs. A more expressive logic leads to a smaller MLAR.
Recall that we choose CTL as arguably the simplest prominent branching-time logic.
For an illustration of this definition, we refer back to \Cref{ex:general_frame}.

\begin{remark}
In the first two cases in the list above, it is also reasonable to consider equivalence-class preserving refinements instead. To do that for $\MLARfin_\cT$ for example, we would consider the class $\cF'_{\cT}$ of 
pairs $(\cS,f)$ of finite abstractions of $\cT$ together with a corresponding abstraction function $f$. Further, we would consider the relation $\abstracts^\equiv$ on this class. 
In fact, all arguments in the sequel also work for this view. While we will not spell this out in detail, we will comment on this at the critical places in the proofs.
Consequently, the bounds we provide also apply to $\Lambda_{\cF'_\cT, \abstracts^\equiv, \mathrm{CTL}}$ and for the analogously defined 
$\Lambda_{\cA'_\cT, \abstracts^\equiv, \mathrm{CTL}}$.
\end{remark}

\subsection{Lower bounds}

The frame $(\cC,\abstracts)$ is always reflexive and transitive as transition systems are refinements of themselves and abstraction functions can be composed to show transitivity.
The modal logic $\mathsf{S4}$ is valid on all reflexive and transitive frames (see \cite{blackburn}) and hence also on all general reflexive and transitive frames. For any class of transition systems $\cC$, we hence have $\mathsf{S4}\subseteq \Lambda_{\cC,\abstracts,\mathrm{CTL}}$.
For other axioms, the validity depends on the class of transition systems we consider. 
More precisely, the structure of the relation $\abstracts$ on $\cC$ allows us to conclude stronger lower bounds.

For our first lower bound of $\mathsf{S4.2}$ for $\MLARfin_\cT$ for any $\cT$, we use that $\mathsf{S4.2}$ is valid on  reflexive, transitive, directed frames (see \cite{blackburn}). A frame $(W,R)$ is directed
if for all $w,v,u\in W$ with $Rwv$ and $Rwu$, there is a $z\in W$ with $Rvz$ and $Ruz$. Showing directedness of $\abstracts$ on $\cF_\cT$ is straightforward (proofs    in Appendix~\ref{app:lower}).
\begin{restatable}{proposition}{lowerfin}
\label{prop:lower_fin}
For any transition system $\cT$, we have $\mathsf{S4.2}\subseteq \MLARfin_\cT$.
\end{restatable}

\noindent
For $\MLARall_\cT$, the lower bound $\mathsf{S4.2.1}$ follows as $\cA_\cT$ contains the  element $\cT$ that is a refinement of all transition systems in $\cA_\cT$. So, directedness and  the validity of ($\mathsf{.2}$) is immediate. The validity of ($\mathsf{.1}$)$=\lbox \ldia p \imp \ldia \lbox p$ on directed frames is equivalent to the existence of a greatest element, which is $\cT$ in this case.

\begin{restatable}{proposition}{lowerall}
\label{prop:lower_all}
For any transition system $\cT$, we have $\mathsf{S4.2.1}\subseteq \MLARall_\cT$.
\end{restatable}

 The axiom $\mathsf{(.1)}$ is  in $\MLAR$ as we can always go to a refinement for which further refinements do not affect the truth of CTL-formulas anymore by making everything except for one path starting from each initial state unreachable. If $p$ corresponds to a CTL-property $\Phi$ s.t. $\lbox\ldia p$ holds at some transition system~$\cT$, then it has to hold at these ``maximal'' refinements and hence $\ldia\lbox p$ is true, too.

\begin{restatable}{proposition}{proplower}
\label{prop:lower}
We have $\mathsf{S4.1}\subseteq \MLAR$.
\end{restatable}

\section{Upper bounds}

In this section, we prove upper bounds for the MLARs. 
An overview of the results can be found in Table \ref{tbl:overview}.
The key technical vehicle for the proofs
are so-called \emph{control statements}. Given a general frame $\GF_{\GFW,\GFR,\cL} = (\GFW,\GFR,\mathrm{Admiss}(\mathcal{L}))$ as in Def. \ref{def:G},
control statements are formulas in the language $\cL$ that exhibit certain  patterns of truth values under the relation $\GFR$.
We summarize the types of control statements introduced in \cite{hamkins2008modal,hamkins2015structural} in Section \ref{sec:control}.
These control statements are sufficient to prove an upper bound of $\mathsf{S4.2}$ in Section \ref{sec:S4.2}.
For the upper bounds $\mathsf{S4.2.1}$ and $\mathsf{S4FPF}$, we develop new control statements in Section \ref{sec:S4.2.1}.

\subsection{Control statements }
\label{sec:control}

The definitions and results given in this subsection are adapted from \cite{hamkins2008modal} and \cite{hamkins2015structural}. 
We define three different types of control statements, namely \emph{pure buttons}, a more general version of them called \emph{pure weak buttons}, and \emph{switches}. Intuitively, a pure button can always be made true by moving to another structure via the relation $\GFR$ and stays true once it is true. In comparison, a pure weak button might eventually have to stay false forever, but also
has to stay true once it is true. Lastly, a switch can always be made true and be made false via $\GFR$.

\vspace{2pt}
\noindent\textbf{Pure buttons and switches.}
Formally, let $\GF$ be the general Kripke frame wrt. some class of structures $\GFW$, relation $\GFR$, and language $\mathcal{L}$ as defined in Def. \ref{def:G} and let $\GFw \in \GFW$.
A sentence $\button \in \mathcal{L}$ is a \emph{pure button} in $\GF$ at $\GFw$ if, for any valuation $V$ with $V(b) = \{\GFv \in \GFW \mid \GFv \sat \button\}$,
we have $(\GF,V),\GFw \Sat \lbox (b \imp \lbox b) $  and $(\GF,V),\GFw \Sat \lbox \ldia  b$.
For an illustration, consider  \Cref{sub@fig:pure_button} and view the root  as the world $c$: From every world reachable from $c$, a world at which the pure button is true is reachable and worlds where the pure button is true can only reach such worlds.
A pure button $\button$ is \emph{pushed} at $\GFu \in \GFW$ with $\GFw \GFR \GFu$ if $(\GF,V),\GFu \Sat b$ and \emph{unpushed} otherwise.

A sentence $\weakButton \in \mathcal{L}$ is a \emph{pure weak button} in $\GF$ at $\GFw$ if, for any valuation $V$ with  $V(b) = \{\GFv \in \GFW \mid \GFv \sat \weakButton\}$,
we have $(\GF,V),\GFw \Sat \lbox (b \imp \lbox b)  $ and $ (\GF,V),\GFw \Sat \ldia b$. So, it is not necessary that the pure weak button can always be made true.
The truth value pattern  is illustrated in  \Cref{sub@fig:pure_weak_button}. A pure weak button $\weakButton$ is \emph{pushed} at $\GFu \in \GFW$ with $\GFw \GFR \GFu$ if $(\GF,V),\GFu \Sat b$ and \emph{unpushed} otherwise.
A pure weak button $\weakButton$ is \emph{intact} at $\GFu \in \GFW$ with $\GFw \GFR \GFu$ if $(\GF,V),\GFu \Sat \ldia b$ and \emph{broken} otherwise.

A sentence $\switch \in \mathcal{L}$ is a \emph{switch} in $\GF$ at $\GFw$ if, for any valuation $V$ with $V(s) = \{\GFv \in \GFW \mid \GFv \sat \switch\}$,
we have $(\GF,V),\GFw \Sat \lbox (\ldia s \land \ldia \lnot s)$.
A switch $\switch$ is \emph{on} at $\GFu \in \GFW$ with $\GFw \GFR \GFu$ if $(\GF,V),\GFu \Sat s$ and \emph{off} otherwise.
\Cref{sub@fig:switch} shows the truth value pattern: Every world can access  worlds at which the switch is on and  off, respectively.

\begin{figure*}[t]
    \begin{subfigure}[b]{0.3\textwidth}
        \centering
        \resizebox{.6\textwidth}{!}{%
            \begin{tikzpicture}[scale=1,auto,node distance=4mm,>=latex]
                \tikzstyle{rect}=[draw=white,rectangle]
                \node[circle, draw=black, fill = black!10!white] (w0) at (0,0) {};
                \node[circle, draw=black, fill = black!80!white] (w1) at (1,.7) {};
                \node[circle, draw=black, fill = black!10!white] (w2) at (-.5,1.4) {};
                \node[circle, draw=black, fill = black!10!white] (w3) at (0,.7) {};
                \node[circle, draw=black, fill = black!80!white] (w4) at (1,1.4) {};
                \node[circle, draw=black, fill = black!80!white] (w5) at (-.5,2.1) {};
                \node[circle, draw=black, fill = black!10!white] (w6) at (.3,1.4) {};
                \node[circle, draw=black, fill = black!80!white] (w7) at (1.5,2.1) {};
                 \node[circle, draw=black, fill = black!80!white] (w8) at (.5,2.1) {};
                \draw
                (w0)     edge[right]                node{} (w1)
                (w0)     edge[above]                node{} (w3)
                (w1)     edge[above]                node{} (w4)
                (w3)     edge[left]     node{} (w2)
                (w3)     edge[right]                node{} (w4)
                (w3)     edge[above]                node{} (w6)
                (w4)     edge[above]    node{} (w7)
                (w6) edge (w8)
                (w6) edge (w5)
                (w2) edge (w5)
                                            ;
            \end{tikzpicture}
        }
        \caption{Pure button (\raisebox{-.5ex}{\scalebox{2}{$\bullet$}}${}=\mathit{true}$).}
        \label{fig:pure_button}
    \end{subfigure}\hfill
    \begin{subfigure}[b]{0.3\textwidth}
        \centering
        \resizebox{.6\textwidth}{!}{%
            \begin{tikzpicture}[scale=1,auto,node distance=4mm,>=latex]
                \tikzstyle{rect}=[draw=white,rectangle]
                \node[circle, draw=black, fill = black!10!white] (w0) at (0,0) {};
                \node[circle, draw=black, fill = black!80!white] (w1) at (1,.7) {};
                \node[circle, draw=black, fill = black!10!white] (w2) at (-.5,1.4) {};
                \node[circle, draw=black, fill = black!10!white] (w3) at (0,.7) {};
                \node[circle, draw=black, fill = black!80!white] (w4) at (1,1.4) {};
                \node[circle, draw=black, fill = black!80!white] (w5) at (-.5,2.1) {};
                \node[circle, draw=black, fill = black!10!white] (w6) at (.3,1.4) {};
                \node[circle, draw=black, fill = black!80!white] (w7) at (1.5,2.1) {};
                 \node[circle, draw=black, fill = black!10!white] (w8) at (.5,2.1) {};
                \draw
                (w0)     edge[right]                node{} (w1)
                (w0)     edge[above]                node{} (w3)
                (w1)     edge[above]                node{} (w4)
                (w3)     edge[left]     node{} (w2)
                (w3)     edge[right]                node{} (w4)
                (w3)     edge[above]                node{} (w6)
                (w4)     edge[above]    node{} (w7)
                (w6) edge (w8)
                (w6) edge (w5)
                (w2) edge (w5)
                                            ;
            \end{tikzpicture}
        }
        \caption{Pure weak button (\raisebox{-.5ex}{\scalebox{2}{$\bullet$}}${}=\mathit{true}$).}
        \label{fig:pure_weak_button}
    \end{subfigure}\hfill
    \begin{subfigure}[b]{0.3\textwidth}
        \centering
        \vspace{-18pt}
       \resizebox{.6\textwidth}{!}{%
            \begin{tikzpicture}[scale=1,auto,node distance=4mm,>=latex]
                \tikzstyle{rect}=[draw=white,rectangle]
                \node[circle, draw=black, fill = black!10!white] (w0) at (0,0) {};
                \node[circle, draw=black, fill = black!80!white] (w1) at (1,.7) {};
                \node[circle, draw=black, fill = black!80!white] (w2) at (-.5,1.4) {};
                \node[circle, draw=black, fill = black!10!white] (w3) at (0,.7) {};
                \node[circle, draw=black, fill = black!80!white] (w4) at (1,1.4) {};
                \node[circle, draw=black, fill = black!80!white] (w5) at (-.5,2.1) {};
                \node[circle, draw=black, fill = black!10!white] (w6) at (.3,1.4) {};
                \node[circle, draw=black, fill = black!80!white] (w7) at (1.5,2.1) {};
                 \node[circle, draw=black, fill = black!10!white] (w8) at (.5,2.1) {};
                  \node (d1) at (0,2.8) {\vdots};
                   \node (d2) at (1,2.8) {\vdots};
                \draw
                (w0)     edge[right]                node{} (w1)
                (w0)     edge[above]                node{} (w3)
                (w1)     edge[above]                node{} (w4)
                (w3)     edge[left]     node{} (w2)
                (w3)     edge[right]                node{} (w4)
                (w3)     edge[above]                node{} (w6)
                (w4)     edge[above]    node{} (w7)
                (w6) edge (w8)
                (w6) edge (w5)
                (w2) edge (w5)
                (w2) edge (w8)
                (w4) edge (w8)
                (w8) edge (d2)
                (w8) edge (d1)
                (w5) edge (d1)
                (w7) edge (d2)
                (w1) edge (w6)
                                            ;
            \end{tikzpicture}
        }
                \caption{Switch (\raisebox{-.5ex}{\scalebox{2}{$\bullet$}}${}=\mathit{true}$).}
        \label{fig:switch}
    \end{subfigure}
    \vspace{-8pt}
    \caption{Illustration of the truth value patterns for different control statements. In all cases, the relation between worlds is  meant to be reflexive and transitive.}
            \vspace{-12pt}
    \label{fig:basic_control_statements}
\end{figure*}
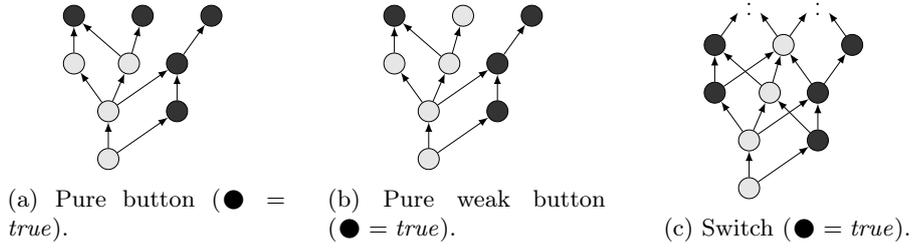

\vspace{2pt}
\noindent\textbf{Independence.}
Different control statements might interfere with one another. For instance, to push a pure button it might be necessary to also push another pure button; however, to prove upper bounds, it is often required to be able to operate control statements  \emph{independently}. Formally, we define this as follows:

Let $n,m \in \mathbb{N}$. Let $\buttonSet = \{\button_i \mid 0 \le i \le n-1\}$ be a set of unpushed pure buttons in $\GF$ at $\GFw$ and let $\switchSet = \{\switch_j \mid 0 \le j \le m-1\}$ be a set of switches in $\GF$ at $\GFw$. We call $\buttonSet \cup \switchSet$ \emph{independent} if for every $I_0 \subseteq I_1 \subseteq \{0, 1, \ldots, n-1\}$ and $J_0, J_1 \subseteq \{0, 1, \ldots, m-1\}$,
and any $V$ with $V(b_i) = \{\GFv \in \GFW \mid \GFv \sat \button_i\}$ and $V(s_j) = \{\GFv \in \GFW \mid \GFv \sat \switch_j\}$,
\begin{align*}
(\GF,V),\GFw \Sat \lbox \big[ &\big( 
    \textstyle\bigwedge\nolimits_{i \in I_0} b_i \land 
    \bigwedge\nolimits_{i \notin I_0} \lnot b_i \land 
    \bigwedge\nolimits_{j \in J_0} s_j \land 
    \bigwedge\nolimits_{j \notin J_0} \lnot s_j 
\big) \\ \imp 
\ldia &\big( 
    \textstyle\bigwedge\nolimits_{i \in I_1} b_i \land 
    \bigwedge\nolimits_{i \notin I_1} \lnot  b_i \land 
    \bigwedge\nolimits_{j \in J_1} s_j \land 
    \bigwedge\nolimits_{j \notin J_1} \lnot s_j 
\big) \big]    .
\end{align*}
This formula says that no matter which set of pure buttons $I_0$ and switches $J_0$ is currently true, there is an accessible world at which an arbitrary other combination of pure buttons $I_1\supseteq I_0$ and switches $J_1$ is true. So, it is possible to change the truth values of any control statement without influencing  the truth values of the others  (besides the fact that pushed pure buttons have to remain true).
Independent families of control statements are a powerful tool to prove upper bounds on the modal logics of model constructions as shown in \cite{hamkins2008modal}:

\begin{theorem}[{\cite[Section 2]{hamkins2008modal}}]
    \label{thm:S4.2}
    Let $\GF$ be the general Kripke frame w.r.t. some $\GFW$, $\GFR$ and $\mathcal{L}$ as in Def. \ref{def:G} and let $\GFw \in \GFW$. If for any $n,m \in \mathbb{N}$, there is a set of unpushed pure buttons $\buttonSet = \{\button_i \mid 0 \le i \le n-1\}$ in $\GF$ at $\GFw$ and a set of switches $\switchSet = \{\switch_j \mid 0 \le j \le m-1\}$ in $\GF$ at $\GFw$ s.t. $\buttonSet \cup \switchSet$ is independent, then $\Lambda_\GF(\GFw) \subseteq \mathsf{S4.2}$.
\end{theorem}
To prove such theorems the notion of \emph{$F$-labeling} is developed in \cite{hamkins2008modal,hamkins2015structural}.
Let $F = (W,\preccurlyeq)$ be a finite Kripke frame with a root $w_0 \in W$. An {$F$-labeling} for $\GFw$ in $\GF$ assigns a sentence $\phi_{w}$ in the language $\mathcal{L}$ to each node $w \in W$ such that
\begin{enumerate}
    \item For all $\GFv \in \GFW$ with $\GFw \GFR \GFv$, there is exactly one $w \in W$ such that $\GFv \sat \phi_{w}$
    \item For all $\GFv \in \GFW$ with $\GFw \GFR \GFv$, if $\GFv \sat \phi_{w}$, then there is an $\GFu \in \GFW$ with $\GFv \GFR \GFu$ and $\GFu \sat \phi_{v}$ if and only if $w \preccurlyeq v$.
    \item $\GFw \sat \phi_{w_0}$.
\end{enumerate}

In other words, an $F$-labeling assigns sentences $\phi_{w}$  to the worlds $w$ of $F$ such that grouping together the structures in $\GFW$ that make the respective formulas $\phi_w$ true results exactly in the structure of the frame $F$.

\begin{lemma}[{\cite[Lem. 9]{hamkins2015structural}}]
    \label{lem:labeling}
Let $F = (W,\preccurlyeq)$ be a finite Kripke frame with a root $w_0 \in W$ and $w \mapsto \phi_{w}$ be an $F$-labeling for $\GFw$ in $\GF$. For any formula $\phi$, if $F, w_0 \nSat \phi$ then $\phi \notin \Lambda_\GF(\GFw)$.
\end{lemma}

\begin{proof}
Let $V$ be a valuation on $F$ such that $(F,V),w_0 \nSat \phi$. Consider the valuation $V'$ on $\GF$ with $V'(p) = \{\GFv \in \GFW: \GFv \sat \bigvee_{w \in V(p)} \phi_{w} \}$. Then, $(\GF, V'),\GFw \nSat \phi$. \qed
\end{proof}

\begin{figure*}[t]
\vspace{-12pt}
    \begin{subfigure}[b]{0.3\textwidth}
        \centering
        \resizebox{.6\textwidth}{!}{%
            \begin{tikzpicture}[scale=1,auto,node distance=4mm,>=latex]
                \tikzstyle{rect}=[draw=white,rectangle]
                \node[circle, draw=black, fill = black!10!white] (w0) at (0,0) {};
                \node[circle, draw=black, fill = black!10!white] (w1) at (.7,0) {};
                \node[circle, draw=black, fill = black!10!white] (w2) at (-.5,.7) {};
                \node[circle, draw=black, fill = black!10!white] (w3) at (-1.2,.7) {};
                \node[circle, draw=black, fill = black!10!white] (w4) at (.5,.7) {};
                \node[circle, draw=black, fill = black!10!white] (w5) at (1.2,.7) {};
                \node[circle, draw=black, fill = black!10!white] (w6) at (0,1.4) {};
                \node[circle, draw=black, fill = black!10!white] (w7) at (.7,1.4) {};
                \draw
                (w0)     edge[bend right]                node{} (w1)
                (w1)     edge[bend right]                node{} (w0)
                (w2)     edge[bend right]                node{} (w3)
                (w3)     edge[bend right]                node{} (w2)
                (w4)     edge[bend right]                node{} (w5)
                (w5)     edge[bend right]                node{} (w4)
                (w6)     edge[bend right]                node{} (w7)
                (w7)     edge[bend right]                node{} (w6)
                (w0)     edge (w2)
                (w0)     edge (w4)
                (w2) edge (w6)
                (w4) edge (w6)
                                
                                            ;
            \end{tikzpicture}
        }
        \caption{Pre-Boolean algebra.}
        \label{fig:pre-boolean}
    \end{subfigure}\hfill
    \begin{subfigure}[b]{0.3\textwidth}
        \centering
         \resizebox{.6\textwidth}{!}{%
            \begin{tikzpicture}[scale=1,auto,node distance=4mm,>=latex]
                \tikzstyle{rect}=[draw=white,rectangle]
                \node[circle, draw=black, fill = black!10!white] (w0) at (0,0) {};
                \node[circle, draw=black, fill = black!10!white] (w1) at (.7,0) {};
                \node[circle, draw=black, fill = black!10!white] (w2) at (-.5,.7) {};
                \node[circle, draw=black, fill = black!10!white] (w3) at (-1.2,.7) {};
                \node[circle, draw=black, fill = black!10!white] (w4) at (.5,.7) {};
                \node[circle, draw=black, fill = black!10!white] (w5) at (1.2,.7) {};
                \node[circle, draw=black, fill = black!10!white] (w6) at (0,1.4) {};
                \node[circle, draw=black, fill = black!10!white] (w7) at (.7,1.4) {};
                 \node[circle, draw=black, fill = black!10!white] (w8) at (0,2.1) {};
                \draw
                (w0)     edge[bend right]                node{} (w1)
                (w1)     edge[bend right]                node{} (w0)
                (w2)     edge[bend right]                node{} (w3)
                (w3)     edge[bend right]                node{} (w2)
                (w4)     edge[bend right]                node{} (w5)
                (w5)     edge[bend right]                node{} (w4)
                (w6)     edge[bend right]                node{} (w7)
                (w7)     edge[bend right]                node{} (w6)
                (w0)     edge (w2)
                (w0)     edge (w4)
                (w2) edge (w6)
                (w4) edge (w6)
                                (w6) edge (w8)
                                            ;
            \end{tikzpicture}
        }
        \caption{Inverted lollipop.}
        \label{fig:lollipop}
    \end{subfigure}\hfill
    \begin{subfigure}[b]{0.3\textwidth}
        \centering
        \vspace{-18pt}
       \resizebox{.9\textwidth}{!}{%
            \begin{tikzpicture}[scale=1,auto,node distance=4mm,>=latex]
                \tikzstyle{rect}=[draw=white,rectangle]
                \node   (w0) at (0,0) {\scriptsize$(?,?)$};
                \node   (w2) at (-.5,1) {\scriptsize$(?,1)$};
                 \node   (w1) at (-1.5,1) {\scriptsize$(0,?)$};
                   \node   (w3) at (.5,1) {\scriptsize$(?,0)$};
                            \node   (w4) at (1.5,1) {\scriptsize$(1,?)$};
                             \node   (w6) at (-.5,2) {\scriptsize$(0,0)$};
                 \node   (w5) at (-1.5,2) {\scriptsize$(0,1)$};
                   \node   (w7) at (.5,2) {\scriptsize$(1,1)$};
                            \node   (w8) at (1.5,2) {\scriptsize$(1,0)$};
              
                              \draw
                                                (w0) edge (w1)      
                                                (w0) edge (w2)       
                                                (w0) edge (w3)       
                                                (w0) edge (w4)       
                                                (w1) edge (w5)       
                                                (w1) edge (w6)       
                                                (w2) edge (w5)       
                                                (w2) edge (w7)       
                                                (w3) edge (w6)       
                                                (w3) edge (w8)       
                                                (w4) edge (w7)       
                                                (w4) edge (w8)           
                                                          ;
            \end{tikzpicture}
        }
                \caption{Partial function poset.}
        \label{fig:FPF}
    \end{subfigure}
                \vspace{-10pt}
    \caption{Illustration of different types of transitive and reflexive frames.}
    \label{fig:structures}
\end{figure*}

\Cref{thm:S4.2} can be shown using $F$-labelings by showing that sufficiently many independent pure buttons and switches over a structure $c$ allow to provide an $F$-labeling
for any finite pre-Boolean algebra $F$. 
A finite Boolean algebra is a partial order isomorphic to a powerset $2^X$ ordered by inclusion $\subseteq$ for some finite set $X$.
A finite pre-Boolean algebra, in turn, is a partial pre-order,  which is obtained by replacing the nodes of a finite Boolean algebra by  clusters of states that can see each other as illustrated in \Cref{fig:pre-boolean}.
Together with the soundness and completeness result stating that the set of modal formulas valid on all finite pre-Boolean algebras is precisely $\mathsf{S4.2}$, this shows \Cref{thm:S4.2} (see \cite{hamkins2008modal,hamkins2015structural}).

\subsection{Upper bounds}

To prove the upper bounds for the modal logics of abstraction refinements,
we will start by utilizing Thm. \ref{thm:S4.2} to prove   $\MLARfin_\cT\subseteq\mathsf{S4.2}$ for some $\cT$.
Afterwards, we prove new results  to obtain a technique to prove  upper bounds of  $\mathsf{S4.2.1}$ and $\mathsf{S4FPF}$ using new control statements. 
The  soundness and completeness results these proofs rely on and the respective control statements are shown in Table~\ref{tbl:completeness}.

   \begin{table}[t]
    \begin{tabular}[h!]{ p{0.08\textwidth} p{0.45\textwidth} p{0.45\textwidth}}
    logic & sound and complete with respect to & control statements for upper bounds \\\hline
     $\mathsf{S4.2}$ & finite pre-Boolean algebras \cite{hamkins2008modal}&   pure buttons, switches \cite{hamkins2008modal} \\
    $\mathsf{S4.2.1}$ & inverted lollipops \cite{inamdar2016modal}&    pure buttons, restr. switches (Thm. \ref{thm:S4.2.1})  \\
    $\mathsf{S4FPF}$ & finite partial function posets (by def.)&  decisions (\Cref{thm:S4FPF})\\
\end{tabular}
\caption{Soundness and completeness results and  control statements.}
\vspace{-12pt}
\label{tbl:completeness}
\end{table}

\noindent\textbf{Upper bound $\mathsf{S4.2}$.}
\label{sec:S4.2}
Using \Cref{thm:S4.2}, we  prove:
\begin{restatable}{theorem}{upperfin}
\label{thm:upper_fin}
    There is a transition system $\cT$ with $\MLARfin_{\cT} \subseteq \mathsf{S4.2}$.
\end{restatable}

\begin{proof}[Proof sketch]
We provide a transition system $\cT$ and an infinite independent set of pure buttons and switches in $G_{\cF_\cT, \abstracts, \mathrm{CTL}}$ at any $\cS$. By Thm. \ref{thm:S4.2}, this implies $\MLARfin_{\cT}(\cS) \subseteq \mathsf{S4.2}$ for all $\cS \in \cF_\cT$. A sketch of $\cT$ is depicted in Fig. \ref{fig:S4.2_main} where we, notably, have different parts $\cT^\button$, $\cT^\switch$ for the different types of control statements. So, an abstraction of $\cT$ can be viewed as independently abstracting $\cT^\beta$ and $\cT^\sigma$, which  guarantees that we do not have a dependence between a pure button and a switch. We will give full details on $\cT^\button$ and the pure buttons, but only sketch the more involved switches (for the full proof, see App.~\ref{app:upperfin}).

\vspace{6pt}
\noindent
\textbf{Pure buttons.} 
We define $\cT^\button = (S, \rightarrow, \{s\}, \{s, f, a\}, L)$ sketched  in Fig. \ref{fig:buttons_S4.2_main}: 
\begin{itemize}
\item
The state space is $S = \{s,f\} \cup \{(i,h,k)\in \mathbb{N}^3 \mid i\geq 2, i\geq h \geq 1\}$. 
\item
The relation $\rightarrow$ is given by 
\begin{itemize}
\item
$s\rightarrow (i,1,k)$ for all $i\geq 2$ and  all $k$,
\item
$(i,h,k)\rightarrow (i,h+1,k)$ for all $i\geq 2$, all $h<i$, and all $k$,
\item
$(i,i,k) \rightarrow f$ for all $i\geq 2$ and all $k$, and 
$f\rightarrow f$.
\end{itemize}
\item
The labeling over atomic propositions $\{s,f,a\}$ is given by $L(s)=\{s\}$, $L(f)=\{f\}$, and $L((i,h,k)) = \{a\}$ for all $i,h,k$.
\end{itemize}
\vspace{-6pt}

\noindent
Given a finite abstraction $\cS$ of $\cT^\beta$ with abstraction function $f$, we denote the set of states  mapped to the same state as  $(i,h,k)$ by $[(i,h,k)]$.
Overloading the notation, we  identify the state $f((i,h,k))$ with this equivalence class $[(i,h,k)]$.

The idea behind the $i$th pure button is roughly to say that some path along $i$ states labelled with $a$ from $s$ to $f$ has been ``isolated'' and not merged with paths of other length.
For any $i \ge 2$, we define the CTL-sentence
\[\button_i = \exists \lnext (\forall \lnext^{i-1} (a \land \forall \lnext f)).\]
In a finite abstraction $\cS$ of $\cT^\beta$, $\button_i$ holds iff there exists $k \geq 1$ such that $a^if^\omega$ is the only trace of paths starting in $[(i,1,k)]$.
 Once the only trace of paths starting in $[(i,1,k)]$ is $a^if^\omega$, this cannot be changed in a finer abstraction of $\cT^\beta$.

Further, whenever $\button_i$ is unpushed at an abstract transition system $\mathcal{S}_0$ then there is an accessible abstraction $\mathcal{S}_1$ of $\cT^\beta$ where $\button_i$ is pushed. In a finite abstraction there is always a path $[(i,1,k)],\ldots,[(i,i,k)]$ for some $k \geq 1$ that only visits infinite equivalence classes. Otherwise we would have infinitely many finite equivalence classes which leads to a contradiction. Thus, it is possible to safely (without pushing other pure buttons) split off the path $(i,1,k),\ldots,(i,i,k)$ by going to a refinement $\mathcal{S}_1$ that is even equivalence class preserving. 
Note that at any finite transition system $\cS$ only finitely many $\button_i$ can be pushed. So, for any $n \in \mathbb{N}$, there is a $I \subset \mathbb{N}$ with $|I| = n$ such that for any $i \in I$, $\button_i$ is unpushed.

\vspace{6pt}
\noindent
\textbf{Switches.}
In $\cT^\sigma$  we add an additional ``dimension'': we not only have infinitely many copies of paths of equal length, but put them into infinitely many ``groups'' of infinite size. Precisely such a group is depicted in \autoref{fig:switches_S4.2_main} for length $3$. Paths within a group are connected via special states allowing them to ``see'' all other paths in that group but themselves. Then, a switch $\sigma_j$ is on whenever there is a group containing paths of length $j$ where exactly one path is isolated. Again, because of the infinitely many copies, we can ensure that a switch can always be turned on and that we also have independence.
\qed
		  \end{proof}

		  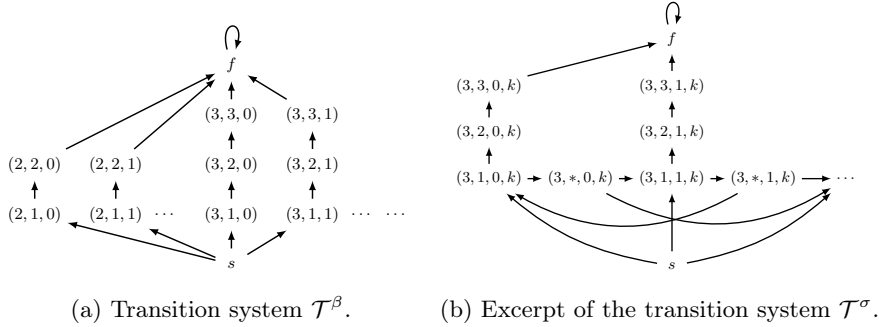
\begin{figure*}[t]
		    \begin{subfigure}[b]{0.45\textwidth}
\centering
    \resizebox{1\textwidth}{!}{%
      \begin{tikzpicture}[scale=1,auto,node distance=3.5mm,>=latex]
        \tikzstyle{rect}=[draw=white,rectangle]

        \node[minimum size=6mm] (a210) {$(2,1,0)$};
         \node[minimum size=6mm,above= of a210] (a220) {$(2,2,0)$};
         
          \node[minimum size=6mm,right=3mm of a210] (a211) {$(2,1,1)$};
         \node[minimum size=6mm,above= of a211] (a221) {$(2,2,1)$};

         \node[minimum size=6mm,right= 0mm of a211] (a212) {$\dots$};

                                \draw[color=black ,  thick] (a210) ->  (a220);
			 \draw[color=black ,  thick] (a211) ->  (a221);

			    \node[minimum size=6mm,right= 3mm of a212] (a310) {$(3,1,0)$};
         \node[minimum size=6mm,above= of a310] (a320) {$(3,2,0)$};
         \node[minimum size=6mm,above= of a320] (a330) {$(3,3,0)$};
         
          \node[minimum size=6mm,right=3mm of a310] (a311) {$(3,1,1)$};
         \node[minimum size=6mm,above= of a311] (a321) {$(3,2,1)$};
         \node[minimum size=6mm,above= of a321] (a331) {$(3,3,1)$};

         \node[minimum size=6mm,right=0mm of a311] (a312) {$\dots$};
          \node[minimum size=6mm,right=0mm of a312] (a313) {$\dots$};
         

           \node[minimum size=6mm,below= of a310] (s) {$s$};
           
             \node[minimum size=6mm,above= of a330] (f) {$f$};
          
             \draw[color=black ,  thick] (a220) ->  (f);
              \draw[color=black ,  thick] (a221) ->  (f);
              
                  \draw[color=black ,  thick] (a330) ->  (f);
              \draw[color=black ,  thick] (a331) ->  (f);

              \draw[color=black ,  thick,->] (f) edge [loop above]   (f);

            \draw[color=black ,  thick] (s) ->  (a310);
              \draw[color=black ,  thick] (s) ->  (a311);

            \draw[color=black ,  thick] (s) ->  (a210);
              \draw[color=black ,  thick] (s) ->  (a211);

                                \draw[color=black ,  thick] (a310) ->  (a320);
                                  \draw[color=black ,  thick] (a320) ->  (a330);
			 \draw[color=black ,  thick] (a311) ->  (a321);
                                  \draw[color=black ,  thick] (a321) ->  (a331);

      \end{tikzpicture}
    }
\caption{Transition system  $\cT^\beta$.}
  \label{fig:buttons_S4.2_main}
  \end{subfigure}
      \begin{subfigure}[b]{0.5\textwidth}
\centering
    \resizebox{.9\textwidth}{!}{%
      \begin{tikzpicture}[scale=1,auto,node distance=3.5mm,>=latex]
        \tikzstyle{rect}=[draw=white,rectangle]

        \node[minimum size=6mm] (a310k) {$(3,1,0,k)$};
         \node[minimum size=6mm,above= of a310k] (a320k) {$(3,2,0,k)$};
                  \node[minimum size=6mm,above= of a320k] (a330k) {$(3,3,0,k)$};

                   \node[minimum size=6mm,right=3mm of a310k] (a3x0k) {$(3,\ast,0,k)$};

                          \node[minimum size=6mm,right=3mm of a3x0k] (a311k) {$(3,1,1,k)$};
         \node[minimum size=6mm,above= of a311k] (a321k) {$(3,2,1,k)$};
                  \node[minimum size=6mm,above= of a321k] (a331k) {$(3,3,1,k)$};
                  
                    \node[minimum size=6mm,right=3mm of a311k] (a3x1k) {$(3,\ast,1,k)$};

         \node[minimum size=6mm,right= 6mm of a3x1k] (a4) {$\dots$};

                    \node[minimum size=6mm,below= 12mm of a311k] (s) {$s$};
            \node[minimum size=6mm,above=  of a331k] (f) {$f$};
            
                          \draw[color=black ,  thick,->] (f) edge [loop above]   (f);

             \draw[color=black ,  thick] (s) edge [bend left=15]  (a310k);
              \draw[color=black ,  thick] (s) edge [bend right=15]  (a4);
             \draw[color=black ,  thick] (s) ->  (a311k);
             \draw[color=black ,  thick] (a310k) ->  (a320k);
             \draw[color=black ,  thick] (a320k) ->  (a330k);
             \draw[color=black ,  thick] (a311k) ->  (a321k);
             \draw[color=black ,  thick] (a321k) ->  (a331k);
                          \draw[color=black ,  thick] (a331k) ->  (f);
             \draw[color=black ,  thick] (a330k) ->  (f);
                          \draw[color=black ,  thick] (a310k) ->  (a3x0k);
                                                    \draw[color=black ,  thick] (a311k) ->  (a3x1k);
                                                  \draw[color=black ,  thick] (a3x0k) ->  (a311k);
                                      \draw[color=black ,  thick] (a3x1k) ->  (a4);
                                      \draw[color=black ,  thick] (a3x1k) edge [bend left]  (a310k);
                                       \draw[color=black ,  thick] (a3x0k) edge [bend right]  (a4);

      \end{tikzpicture}
    }
\caption{Excerpt of the transition system  $\cT^\sigma$.}
  \label{fig:switches_S4.2_main}
  \end{subfigure}
  \caption{The components of the transition system $\cT$ in the proof of \Cref{thm:upper_fin} sharing the initial state $s$.}
          \vspace{-12pt}
\label{fig:S4.2_main}
\end{figure*}

\begin{example}
By providing an upper bound for $\MLARfin_{\cT} $, we show that all formulas $\varphi$ not in $ \mathsf{S4.2}$ are falsifiable in the sense that there is a transition system $\cT$, a finite abstraction $\cS$, and CTL formulas $P$ for each atomic proposition $p$ of $\varphi$ such that $\varphi$ does not hold at $\cS$ in $\cF_\cT$ when making each $p$ true at all transition systems in $\cF_\cT$ satisfying the corresponding $P$.
The proof provides the necessary transition system and abstractions:
Consider, e.g., the formula $(\mathsf{.1})=\lbox \ldia p \imp \ldia \lbox p\not \in \mathsf{S4.2}$. It can be falsified on a pre-Boolean algebra $F$ with just one cluster of two states. A single switch is then sufficient to obtain an $F$-labelling.
Now, letting $\cS$ be any finite abstraction of the transition system $\cT$ constructed in the proof of \Cref{thm:upper_fin} and making $p$ true at all transition systems that satisfy one of the switches $\sigma$, we have that $\lbox \ldia p$ holds at $\cS$, but $\ldia \lbox p$ does not hold   by the definition of a switch.
\end{example}

\noindent\textbf{Upper bound $\mathsf{S4.2.1}$.}
\label{sec:S4.2.1}
The modal logic $\mathsf{S4.2.1}$ is the set of formulas valid on all \emph{inverted lollipops} as shown in \cite{inamdar2016modal}.
An inverted lollipop is a pre-Boolean algebra equipped with an additional single top element that is larger than any element in the pre-Boolean algebra as illustrated in \Cref{fig:lollipop}.
We introduce new control statements, called \emph{$\Button$-restricted switches} for a pure button $\Button$, to make use of this characterization result to prove the upper bound $\mathsf{S4.2.1}$ for $\MLARall_\cT$ for some $\cT$ and for~$\MLAR$. Intuitively, $\Button$-restricted switches differ from regular switches by allowing them to break similarly to how pure weak buttons break by pushing $\Button$.

Formally, let $\GF$ be the general Kripke frame with respect to some $\GFW$, $\GFR$ and $\mathcal{L}$ as defined in \autoref{def:G} and let $\GFw \in \GFW$.
Let $\Button$ be an unpushed pure button in $\GF$ at $\GFw$. A sentence $\switchB \in \mathcal{L}$ is a \emph{$\Button$-restricted switch} in $\GF$ at $\GFw$ if,
for any valuation $V$ with $V(s) = \{\GFv \in \GFW \mid \GFv \sat \switchB\}$ and $V(B) = \{\GFv \in \GFW \mid \GFv \sat \Button\}$,
we have $(\GF,V),\GFw \Sat \lbox (\lnot B \imp (\ldia (s \land \lnot B) \land \ldia (\lnot s \land \lnot B)))$.
A $\Button$-restricted switch $\switchB$ is \emph{on} at $\GFu \in \GFW$ with $\GFw \GFR \GFu$ if $(\GF,V),\GFu \Sat s$ and \emph{off} otherwise.
Intuitively, it must be possible to turn a $\Button$-restricted switch on and off as long as $\Button$ is not true yet. As soon as $\Button$ is true, the truth value of the restricted switch does not matter. This pattern is illustrated in  \Cref{sub@fig:restricted_switch} where the restricted switch  is true at worlds filled in black   and squares represent worlds where the restricting pure button is true.
Further, a $\Button$-restricted switch $\switchB$ is \emph{intact} at $\GFu \in \GFW$ with $\GFw \GFR \GFu$ if $(\GF,V),\GFu \Sat \lnot B$ and \emph{broken} otherwise.

We again need a notion of independence between different $\Button$-restricted switches; also in combination with pure buttons. As $\Button$-restricted switches are only meaningful as long as they are intact we define independence only until $\Button$ is pushed.

Let $n,m \in \mathbb{N}$. Let $\buttonSet = \{\button_i \mid 0 \le i \le n-1\}\cup \{\Button \}$ be a set of unpushed pure buttons in $\GF$ at $\GFw$ and let $\switchSet = \{\switchB_j \mid 0 \le j \le m-1\}$ be a set of $\Button$-restricted switches in $\GF$ at $\GFw$. We call $\buttonSet \cup \switchSet$ \emph{independent until $\Button$} if for any $I_0 \subseteq I_1 \subseteq \{0, 1, \ldots, n-1\}$ and $J_0, J_1 \subseteq \{0, 1, \ldots, m-1\}$,
and any $V$ with $V(b_i) = \{\GFv \in \GFW \mid \GFv \sat \button_i\}$, $V(s_j) = \{\GFv \in \GFW \mid \GFv \sat \switchB_j\}$ and $V(B) = \{\GFv \in \GFW \mid \GFv \sat \Button\}$,
\begin{align*}
(\GF,V),\GFw \Sat \lbox \big[ &\big( 
   \textstyle \bigwedge\nolimits_{i \in I_0} b_i \land 
    \bigwedge\nolimits_{i \notin I_0} \lnot b_i \land 
    \bigwedge\nolimits_{j \in J_0} s_j \land 
    \bigwedge\nolimits_{j \notin J_0} \lnot s_j \land
    \lnot B
\big) \\ \imp 
\ldia &\big( 
  \textstyle  \bigwedge\nolimits_{i \in I_1} b_i \land 
    \bigwedge\nolimits_{i \notin I_1} \lnot  b_i \land 
    \bigwedge\nolimits_{j \in J_1} s_j \land 
    \bigwedge\nolimits_{j \notin J_1} \lnot s_j  \land
    \lnot B
\big) \big]    .
\end{align*}
So, we require that it is possible to change truth values of switches and push the pure buttons independently (without pushing $\Button$) only as long as $\Button$ is not true. Once $\Button$ is true, the truth values of the control statements are irrelevant.

\begin{theorem}
\label{thm:S4.2.1}
Let $\GF$ be the general frame w.r.t. some $\GFW$, $\GFR$ and $\mathcal{L}$ as in Def.~\ref{def:G} and let $\GFw \in \GFW$. If for any $n,m \in \mathbb{N}$, there are sets of unpushed pure buttons $\buttonSet = \{\button_i \mid 0 \le i \le n-1\} \cup \{\Button\}$ in $\GF$ at $\GFw$ and  of $\Button$-restricted switches $\switchSet = \{\switchB_j \mid 0 \le j \le m-1\}$ in $\GF$ at $\GFw$ s.t. $\buttonSet \cup \switchSet$ is independent until $\Button$, then $\Lambda_\GF(\GFw) \subseteq \mathsf{S4.2.1}$.
\end{theorem}

\begin{figure*}[t]
    \begin{subfigure}[b]{0.51\textwidth}
        \centering
       \resizebox{.33\textwidth}{!}{%
            \begin{tikzpicture}[scale=1,auto,node distance=4mm,>=latex]
                \tikzstyle{rect}=[draw=black,rectangle,minimum size=9pt]
                \node[circle, draw=black, fill = black!80!white] (w2) at (-.5,1.4) {};
                \node[circle, draw=black, fill = black!10!white] (w3) at (0,.7) {};
                \node[circle, draw=black, fill = black!80!white] (w4) at (1,1.4) {};
                \node[circle, draw=black, fill = black!80!white] (w5) at (-.5,2.1) {};
                \node[circle, draw=black, fill = black!10!white] (w6) at (.3,1.4) {};
                \node[circle, draw=black, fill = black!80!white] (w7) at (1.5,2.1) {};
                 \node[circle, draw=black, fill = black!10!white] (w8) at (.5,2.1) {};
                  \node[rect, draw=black, fill = black!10!white] (w9) at (.5,3.2) {};
                                    \node[rect, draw=black, fill = black!80!white] (w10) at (-.5,3.2) {};
                                    \node[rect, draw=black, fill = black!80!white] (w11) at (1.5,3.2) {};
                                    \node[rect, draw=black, fill = black!80!white] (w12) at (1,3.5) {};

                  \node (d1) at (0,2.8) {\vdots};
                   \node (d2) at (1,2.8) {\vdots};
                \draw
                (w3)     edge[left]     node{} (w2)
                (w3)     edge[right]                node{} (w4)
                (w3)     edge[above]                node{} (w6)
                (w4)     edge[above]    node{} (w7)
                (w6) edge (w8)
                (w6) edge (w5)
                (w2) edge (w5)
                (w2) edge (w8)
                (w4) edge (w8)
                (w8) edge (d2)
                (w8) edge (d1)
                (w5) edge (d1)
                (w7) edge (d2)
                (d1) edge (w9)
                (d1) edge (w10)
                (d2) edge (w9)
                (d2) edge (w11)
                (w9) edge (w12)
                (w11) edge (w12)
                                            ;
            \end{tikzpicture}
        }
                \caption{Restricted switch: true at worlds depicted as  \raisebox{-.5ex}{\scalebox{2}{$\bullet$}},{\large\textcolor{black}{\raisebox{-.2ex}{$\blacksquare$}}} and  restricted by a pure button true at worlds depicted as {{\large\textcolor{black}{\raisebox{-.2ex}{$\blacksquare$}}},{\large\textcolor{black!20!white}{\raisebox{-.2ex}{$\blacksquare$}}}}.}
        \label{fig:restricted_switch}
    \end{subfigure}\hfill
    \begin{subfigure}[b]{0.47\textwidth}
        \centering
        \resizebox{.4\textwidth}{!}{%
            \begin{tikzpicture}[scale=1,auto,node distance=4mm,>=latex]
                \tikzstyle{rect}=[draw=white,rectangle,minimum size=9pt]
                \node[circle, draw=black, fill = black!10!white] (w0) at (0,0) {};
                \node[circle, draw=black, fill = black!80!white] (w1) at (1,.7) {};
                \node[rect, draw=black, fill = red!80!white] (w2) at (-.5,1.4) {};
                \node[circle, draw=black, fill = black!10!white] (w3) at (0,.7) {};
                \node[circle, draw=black, fill = black!80!white] (w4) at (1,1.4) {};
                \node[rect, draw=black, fill = red!80!white] (w5) at (-.5,2.1) {};
                \node[rect, draw=black, fill = red!80!white]  (w6) at (.3,1.4) {};
                \node[circle, draw=black, fill = black!80!white] (w7) at (1.5,2.1) {};
                 \node[rect, draw=black, fill = red!80!white] (w8) at (.5,2.1) {};
                \draw
                (w0)     edge[right]                node{} (w1)
                (w0)     edge[above]                node{} (w3)
                (w1)     edge[above]                node{} (w4)
                (w3)     edge[left]     node{} (w2)
                (w3)     edge[right]                node{} (w4)
                (w3)     edge[above]                node{} (w6)
                (w4)     edge[above]    node{} (w7)
                (w6) edge (w8)
                (w6) edge (w5)
                (w2) edge (w5)
                                            ;
            \end{tikzpicture}
        }
        \caption{Decision consisting of two pure weak buttons true at worlds depicted as \raisebox{-.4ex}{\scalebox{2}{$\bullet$}} and  {\large\textcolor{red}{\raisebox{-.2ex}{$\blacksquare$}}}, respectively.}
        \label{fig:decision}
    \end{subfigure}
      \vspace{-4pt}
    \caption{Illustration of the truth value patterns for the new control statements. In both cases, the relation between worlds is the reflexive transitive closure of the relation indicated by the arrows.}
            \vspace{-12pt}
    \label{fig:new_control_statements}
\end{figure*}
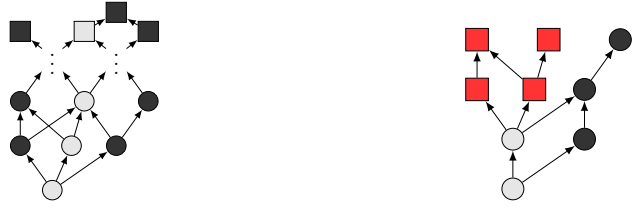

\begin{proof}
Without loss of generality, we can assume all $\Button$-restricted switches are off at $\GFw$.
Since $\mathsf{S4.2.1}$ is characterized by the class of finite inverted lollipops \cite{inamdar2016modal}, if a formula $\phi \notin \mathsf{S4.2.1}$, then there is an inverted lollipop $L = (W,\preccurlyeq)$, a valuation $V$ on $L$ and a initial state $w_0 \in W$ such that $(L,V),w_0 \nSat \phi$. $L$ has an underlying pre-Boolean algebra containing $n$ atom clusters for some $n$. Furthermore, by copying states, we can assume that each cluster consists of $2^m$ states for some $m$.
Thus, we can write $W = \powerset{n} \times \powerset{m} \cup \{\star\}$ where $\star$ is greater than any other element, i.e., for all $(I,J) \in \powerset{n} \times \powerset{m}$, $(I,J) \preccurlyeq \star$, and for all $(I,J),(I',J') \in \powerset{n} \times \powerset{m}$, $(I,J) \preccurlyeq (I',J')$ if and only if $I \subseteq I'$.
We define formulas $\phi_{w} \in \mathcal{L}$ for every $I \subseteq \{0, 1, \ldots, n-1\}$ and $J \subseteq \{0, 1, \ldots, m-1\}$
\[\phi_{(I,J)} = 
( \textstyle\bigwedge\nolimits_{i \in I} \button_i ) \land 
(\textstyle \bigwedge\nolimits_{j \in J} \switch_j ) \land \lnot \Button \quad \text{ and } \quad
\phi_{\star} = \Button.\]

\noindent We check that $w \mapsto \phi_{w}$ satisfies the conditions of an $L$-labeling for $\GFw$ in $\GF$:
\begin{enumerate}
    \item Clearly, every $\GFv \in \GFW$ satisfies exactly one of these formulas.
    \item Since the pure buttons and $\Button$-restricted switches are independent at $\GFw$, for every $I, I' \subseteq \{0, 1, \ldots, n-1\}$, $J,J' \subseteq \{0, 1, \ldots, m-1\}$ and $\GFv \in \GFW$, if $\GFv \sat \phi_{(I,J)}$ then there is an $\GFu \in \GFW$ with $\GFv \GFR \GFu$ and $\GFu \sat \phi_{(I',J')}$ iff $I \subseteq I'$. Further, by pushing $\Button$ there is always an extension where $\phi_{\star}$ holds and as $\Button$ is a pure button every further extension also satisfies $\phi_{\star}$.
    \item As all pure buttons are unpushed and the $\Button$-restricted switches are all off at $\GFw$., $\GFw \sat \phi_{(\emptyset,\emptyset)}$ 
\end{enumerate}
By \autoref{lem:labeling}, since $(L,V),w_0 \nSat \phi$, we have $\phi \notin \Lambda_\GF(\GFw)$. So, for any formula $\phi \notin \mathsf{S4.2.1}$, we have $\phi \notin \Lambda_\GF(\GFw)$ and thus, $\Lambda_\GF(\GFw) \subseteq \mathsf{S4.2.1}$.
\qed
\end{proof}

\begin{corollary}
\label{cor:S4.2.1}
Let $\GF$ be the general  frame w.r.t. some $\GFW$, $\GFR$ and $\mathcal{L}$ as in Def. \ref{def:G}. If for any $n,m \in \mathbb{N}$, there are a state $\GFw \in \GFW$ and  sets of unpushed pure buttons $\buttonSet = \{\button_i \mid 0 \le i \le n-1\} \cup \{\Button\}$ and of $\Button$-restricted switches $\switchSet = \{\switchB_j \mid 0 \le j \le m-1\}$ in $\GF$ at $\GFw$ s.t. $\buttonSet \cup \switchSet$ is independent until $\Button$, then $\Lambda_\GF \subseteq \mathsf{S4.2.1}$.
\end{corollary}

\begin{proof}
    By the same argument, for any $n,m \in \mathbb{N}$, there exists a state $\GFw \in \GFW$ such that $w \mapsto \phi_{w}$ is a $L$-labeling for $\GFw$ in $\GF$. By \autoref{lem:labeling}, since $(L,V),w_0 \nSat \phi$, we have $\phi \notin \Lambda_\GF(\GFw)$ and thus, $\phi \notin \Lambda_\GF$. So, for any formula $\phi \notin \mathsf{S4.2.1}$, we have $\phi \notin \Lambda_\GF$ and thus, $\Lambda_\GF \subseteq \mathsf{S4.2.1}$.
    \qed
\end{proof}

\noindent
Using  Thm. \ref{thm:S4.2.1}, we now prove the upper bound $\mathsf{S4.2.1}$ for $\MLARall_{\cT}$ for some $\cT$.

\begin{restatable}{theorem}{upperboundMLARall}
\label{thm:upper_all}\label{thm:upperMLARall}
    There is a transition system $\cT$ with $\MLARall_{\cT} \subseteq \mathsf{S4.2.1}$.
\end{restatable}

\begin{proof}[Proof sketch]
    We provide transition systems $\cT$ and $\cS \in \cA_\cT$ with arbitrarily many independent unpushed pure buttons and $\Button$-restricted switches in $G_{\cA_\cT, \abstracts, \mathrm{CTL}}$ at $\cS$. For each pure button, we have a pair of states in $\cT$ that is collapsed in $\cS$. By splitting it apart, the respective pure button gets pushed. Similarly, we use  infinitely many collapsed triples of states for each $\Button$-restricted switch. If one of the triples is refined to two states, the switch is on; otherwise, it is off.
    So, a switch can be turned on by 
    splitting a  fully collapsed triple into two parts and off by refining all triples split into two to three states.  
    As long as infinitely many triples are still collapsed to a single state, the switch can hence be turned on and off at will.
    The restricting  pure button $\Button$ using an infinite chain of states from which the triples are reachable now expresses that  all but finitely many triples are at least partially refined.
    As long as $\Button$ is not true, the switch is intact.  The full proof is  in Appendix \ref{app:MLARall}.
    \qed
\end{proof}

To also prove  $\MLAR \subseteq \mathsf{S4.2.1}$, we are not able to provide a single transition system $\cS$ with $\MLAR(\cS)\subseteq \mathsf{S4.2.1}$. Instead we rely on 
Cor. \ref{cor:S4.2.1} and provide a family $\{\cS_{n,m}\}_{n,m \in \mathbb{N}}$ of transition systems with $n+1$ unpushed pure buttons (including some  $\Button$) and $m$ $\Button$-restricted switches. 
The  proof is given in   \Cref{app:upper}.
\begin{restatable}{theorem}{thmupper}
\label{thm:upper}
    $\MLAR \subseteq \mathsf{S4.2.1}$.
\end{restatable}

\noindent\textbf{Upper bound $\mathsf{S4FPF}$.}
\label{sec:FPF}
To further improve the upper bound for $\MLAR$, we define the logic $\mathsf{S4FPF}$ semantically as the logic of  finite partial function posets.

\begin{definition}
For any $n \in \mathbb{N}$, let $F$ be the set of partial functions $f \colon \{0, 1, \ldots, n-1\} \to\{0, 1\}$ and $\preccurlyeq$ the partial order on $F$ with $g \preccurlyeq h$ iff $g = h{\restriction_{\operatorname{dom}(g)}}$ where $\operatorname{dom}(g)$ is the domain of $g$ for all $f,g$. We call  $ (F,\preccurlyeq)$ a finite partial function (FPF) poset on $n$ elements.
The modal theory $\mathsf{S4FPF}$ is  the set of all formulas valid on the class of all FPF posets.
\end{definition}
A FPF poset with $n=2$ is illustrated in \Cref{fig:FPF}. There, a partial function $f\colon \{0,1\} \to \{0,1\}$ is  depicted as pair $(f(0),f(1))$ with $?$ indicating that the value is not defined.
Before we proceed, we relate the logic $\mathsf{S4FPF}$  to well-known logics. 
Grzegorczyk’s logic 
$
\mathsf{Grz}= \mathsf{K} + \lbox(\lbox(p\to \lbox p) \to p)\to p
$
is valid on all reflexive, transitive, weakly conversely well-founded frames (also called Noetherian partial orders), i.e., frames in which all infinite paths take only one self-loop from some point on  (see, e.g., \cite[Section 3.8]{chagrov1997modal}). As FPF posets have this property, $\mathsf{S4FPF}$ is an extension of $\mathsf{Grz}\supseteq \mathsf{S4.1}$. As $\mathsf{Grz}\not\subseteq \mathsf{S4.2.1}$, we conclude $\mathsf{S4FPF}\not\subseteq \mathsf{S4.2.1}$ (see, e.g., \cite{DBLP:journals/aml/BezhanishviliBL21} for the relation of $\mathsf{Grz}$ to other extensions of $\mathsf{S4}$).
Further, as FPF posets are not directed, the axiom $\mathsf{(.2)} = \ldia \lbox p \to \lbox\ldia p$ is not included in $\mathsf{S4FPF}$. So:

\begin{proposition}
\label{prop:s4fpf}
We have $\mathsf{Grz} \subseteq \mathsf{S4FPF}$,  $\mathsf{S4FPF}\not \subseteq \mathsf{S4.2.1}$, and $\mathsf{S4.2.1}\not \subseteq \mathsf{S4FPF}$.
\end{proposition}

To prove an upper bound of $\mathsf{S4FPF}$, we introduce \emph{decisions} as new control statements. A decision consists of two mutually exclusive pure weak buttons. 
So, there is a dependence: by pushing one of them,  we break the other one.
Formally, let $\GF$ be the general  frame with respect to some $\GFW$, $\GFR$ and $\mathcal{L}$ as defined in Def. \ref{def:G} and let $\GFw \in \GFW$.
A pair of unpushed pure weak buttons $(\weakButton, \weakButtonAlt)$ is a \emph{decision} in $\GF$ at $\GFw$ if $\weakButton \lor \weakButtonAlt$ is an unpushed pure button in $\GF$ at $\GFw$ and for any $V$ with $V(l) = \{\GFv \in \GFW: \GFv \sat \weakButton \}$ and $V(r) = \{\GFv \in \GFW: \GFv \sat \weakButtonAlt \}$,
$(\GF,V),\GFw \Sat \lbox (\lnot l \lor \lnot r)$ and $(\GF,V),\GFw \Sat \lbox ((\ldia l \land \ldia r) \lor l \lor r)$.
A decision $(\weakButton, \weakButtonAlt)$ is \emph{pushed} or \emph{decided} at $\GFu \in \GFW$ with $\GFw \GFR \GFu$ if $(\GF,V),\GFu \Sat l \lor r$ and \emph{unpushed} or \emph{undecided} otherwise. An illustration  can be found in \Cref{sub@fig:decision}.

While decisions are  internally dependent, we  again require independence among different decisions.
Let $n \in \mathbb{N}$. Let $\buttonSet = \{(\weakButton_i, \weakButtonAlt_i) \mid 0 \le i \le n-1\}$ be a set of unpushed decisions in $\GF$ at $\GFw$. We call $\buttonSet$ independent if for every $I_0 \subseteq I_1 \subseteq \{0, 1, \ldots, n-1\}$ and $J_0 \subseteq J_1 \subseteq \{0, 1, \ldots, n-1\} \setminus I_1$,
and any $V$ with $V(l_i) = \{\GFv \in \GFW \mid \GFv \sat \weakButton_i\}$ and $V(r_i) = \{\GFv \in \GFW \mid \GFv \sat \weakButtonAlt_i\}$,
\begin{align*}
(\GF,V),\GFw \Sat \lbox \big[ &\big(  
  \textstyle  \bigwedge_{i \in I_0} l_i \land 
    \bigwedge_{i \notin I_0} \lnot l_i \land 
    \bigwedge_{i \in J_0} r_i \land 
    \bigwedge_{i \notin J_0} \lnot r_i
\big) \\ \imp 
\ldia &\big(  
 \textstyle   \bigwedge_{i \in I_1} l_i \land 
    \bigwedge_{i \notin I_1} \lnot  l_i \land 
    \bigwedge_{i \in J_1} r_i \land 
    \bigwedge_{i \notin J_1} \lnot r_i
\big) \big]    .
\end{align*}

\begin{restatable}{theorem}{thmFPF}
    \label{thm:S4FPF}
    Let $\GF$ be the general Kripke frame w.r.t. some $\GFW$, $\GFR$ and $\mathcal{L}$ as in Def. \ref{def:G} and let $\GFw \in \GFW$. If for any $n \in \mathbb{N}$, there is a set of unpushed decisions $\buttonSet = \{(\weakButton_i, \weakButtonAlt_i) \mid 0 \le i \le n-1\}$ in $\GF$ at $\GFw$ s.t. $\buttonSet$ is independent, then $\Lambda_\GF(\GFw) \subseteq \mathsf{S4FPF}$.
\end{restatable}
\noindent
The proof  works analogously to the proof of Thm. \ref{thm:S4.2.1} and is given in App. \ref{app:FPF}.

\begin{theorem}
\label{thm:upper_fpf}
    $\MLAR \subseteq \mathsf{S4FPF}$.
\end{theorem}

\begin{proof}
We will show that there is a transition system $\cS$ such that $\MLAR(\cS) \subseteq \mathsf{S4FPF}$. By \autoref{thm:S4FPF}, it is sufficient provide a independent set of unpushed decisions $\buttonSet = \{(\weakButton_i, \weakButtonAlt_i) \mid 0 \le i \le n-1\}$ in $G_{\mathfrak{A}, \abstracts, \mathrm{CTL}}$ at $\cS$ for any $n \in \mathbb{N}$.

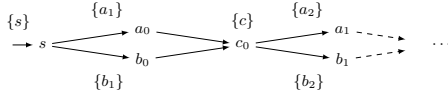
\begin{figure*}[t]
    \begin{center}
    \resizebox{.5\textwidth}{!}{%
        \begin{tikzpicture}[scale=1,auto,node distance=4mm,>=latex]
            \tikzstyle{rect}=[draw=white,rectangle]
            \node[initial left, label=125:$\{s\}$]         (s) at ( 0, 0)  {$s$};
            \node[label=155:{$\{a_1\}$}]                   (a1) at (2, .3)  {$a_0$};
            \node[label=215:{$\{b_1\}$}]                   (b1) at (2, -.3) {$b_0$};
            \node[label=90 :{$\{c\}$}]                   (c1) at (4, 0)      {$c_0$};
            \node[label=155:{$\{a_2\}$}]                   (a2) at (6, .3)  {$a_1$};
            \node[label=215:{$\{b_2\}$}]                   (b2) at (6, -.3) {$b_1$};
            \node[]                                               (c2) at (8,0)      {$\cdots$};
            \draw
            (s)     edge[right]                                node{} (a1)
            (s)     edge[right]                                node{} (b1)
            (a1)    edge[right]                                node{} (c1)
            (b1)    edge[right]                                node{} (c1)
            (c1)    edge[right]                                node{} (a2)
            (c1)    edge[right]                                node{} (b2)
            (a2)    edge[dashed, right, shorten >=12pt]        node{} (c2)
            (b2)    edge[dashed, right, shorten >=12pt]        node{} (c2)
            ;  
        \end{tikzpicture}
    }
    \end{center}
    \vspace{-16pt}
    \caption{The transition system $\cS$ in the proof of \Cref{thm:upper_fpf}.}
        \vspace{-12pt}
    \label{fig:decisions}
\end{figure*}

We define $\cS = (S, \rightarrow, \{s\}, \AP, L)$ depicted also in Figure \ref{fig:decisions} where
\begin{itemize}
    \item
    The state space is $S = \{s\} \cup \{a_i \mid i \in \mathbb{N}\} \cup \{b_i \mid i \in \mathbb{N}\} \cup \{c_i \mid i \in \mathbb{N}\}$.
    \item 
    $s \rightarrow a_1$,
    $s \rightarrow b_1$,
    $a_i \rightarrow c_i$,    $b_i \rightarrow c_i$,    $c_i \rightarrow a_{i+1}$,    
    $c_i \rightarrow b_{i+1} $ for all~$i$.
    \item
    The set of atomic propositions is $\AP = \{s, c\} \cup \{a_i \mid i \in \mathbb{N}\} \cup \{b_i \mid i \in \mathbb{N}$\}.
    \item
    The labeling over $\AP$ is given by $L(s) = \{s\}$, $L(a_i) = \{a_i\}$ for all $i$, $L(b_i) = \{b_i\}$ for all $i$ and $L(c_i) = \{c\}$ for all $i$.
\end{itemize}
\noindent
\textbf{Decisions.} For any $i$, we define the CTL-sentences
$\weakButton_i = \forall \globally \neg a_i $ and $\weakButtonAlt_i = \forall \globally \neg b_i $ expressing that $a_i$ or $b_i$ is not reachable anymore.
They are pure weak buttons  at $\cS_{n,m}$ since it is possible to simply remove $a_i$ or $b_i$, respectively, by going to a refinement where the respective state is not reachable from the initial state. Further, if $\weakButton_i$ gets pushed, $\weakButtonAlt_i$ breaks and vice-versa, meaning that $(\weakButton_i, \weakButtonAlt_i)$ is a decision. Independence is clear as we use different labels.
\qed
\end{proof}

\noindent
In conclusion, we showed that $\mathsf{S4.1}\subseteq \MLAR \subseteq \mathsf{S4.2.1}\cap \mathsf{S4FPF}$.
As shown in Prop. \ref{prop:s4fpf}, $\mathsf{S4.2.1}\cap \mathsf{S4FPF}$ is a stronger bound than $\mathsf{S4FPF}$ or $\mathsf{S4.2.1}$ on their own.

\section{Conclusion}

We defined the modal logics of abstraction refinement in three settings: $\MLARfin_\cT$ and $\MLARall_\cT$ are defined on the frame of all finite and all abstractions of a transition system $\cT$, respectively, while  $\MLAR$ is defined on the class of all transition systems.
In the first two cases, we proved matching upper and lower bounds of $\mathsf{S4.2}$ and $\mathsf{S4.2.1}$, respectively.
These results also unveil the computational complexity of the respective MLARs as
checking validity and satisfiability for both $\mathsf{S4.2}$ and $\mathsf{S4.2.1}$ are PSPACE-complete \cite{shapirovsky2004pspace}.
Determining the precise logic  $\MLAR$ for which we know $\mathsf{S4.1}\subseteq\MLAR\subseteq\mathsf{S4.2.1}\cap \mathsf{S4FPF}$ remains as future work.
The fact that $\mathsf{S4FPF}$ is defined semantically raises the question whether this logic can 
be finitely axiomatized.

The MLARs of course depend on the used notion of abstraction and the use of CTL to describe properties.
Some observations are immediate: For example, for any logic at least as expressive as CTL such as CTL$^\ast$, the upper bounds shown here hold as well. For the proof of the upper bounds, we only need that the respective control statements are expressible in the logic. The lower bounds, on the other hand, only depend on the structure of the refinement relation and hence apply no matter which logic is used instead of CTL. 
Regarding the notion of abstraction, we considered the well-established notion of \emph{existential} abstraction \cite{DBLP:reference/mc/DamsG18} here. 
The investigation of different abstraction and refinement relations such as predicate abstraction might be an interesting direction for future research.

\noindent\textbf{Discussion of the results.}
Seminal abstraction refinement paradigms such as CEGAR are limited to  branching-time properties expressed in syntactic fragments of logics like the universally quantified fragment of CTL and CTL$^\ast$ (subsuming  implicitly universally quantified LTL-properties).
One of the goals of investigating the MLAR is to determine to which extent such abstraction refinement techniques can be extended to more general branching-time properties.

To some degree, the results presented in this paper are negative: 
There are transition systems for which $\MLARfin_\cT$ and $\MLARall_\cT$ coincide with the lower bounds obtained from the structure of the refinement relation.
So, the evolution of CTL-expressible properties is governed by the same general principles as arbitrary properties.
The only principles that can be exploited in general are the validity of $ \mathsf{(.2)}=\ldia \lbox p \imp \lbox \ldia p$  and  $ \mathsf{(.1)}=\lbox \ldia p \imp \ldia \lbox p$ 
in the case of $\MLARall_\cT$. 
Nevertheless, some potentially useful, although simple, observations are possible:
From $\mathsf{(.2)}$, the formula $\ldia \lbox p \land \ldia \lbox q \to \ldia \lbox (p\land q)$ follows. So, given an abstraction $\cS$ of a system  $\cT$,
 to show that a conjunction $\Phi\land \Psi$ of properties holds in $\cT$, it is sufficient to find separate refinements $\cS_1$ and $\cS_2$ of $\cS$ such that $\Box \Phi$ holds in $\cS_1$ and $\Box \Psi$ holds in $\cS_2$.
Formula  $\mathsf{(.1)}$ means that it is sufficient to show  a property can always be made true by refinement in order to show that it holds in $\cT$. 

The results, however, do not mean that, for any transition system $\cT$ and any abstraction $\cS$, the logics $\MLARfin_\cT(\cS)$ and $\MLARall_\cT(\cS)$ coincide with the lower bounds.
This opens up a path for future investigations on conditions on systems $\cT$ and abstractions $\cS$ under which more modal principles are valid, which in turn could be exploited in 
 abstraction refinement-based verification methods.
 
 Further, the key concept of control statements is worth exploring further. 
On the one hand, identifying which kind of CTL-formulas can act as which kind of control statements could be useful to tailor abstraction refinement  to certain types of specifications.
In particular, investigating when a CTL-formula (e.g., in a certain syntactic fragment) acts as a pure button is an interesting direction for future research
as pure buttons have the desirable property that it is sufficient to find a refinement in which they are true in order to conclude that they are true in the underlying system model.
On the other hand, the new generic results on how to use restricted switches and decisions
to prove upper bounds of $\mathsf{S4.2.1}$ and $ \mathsf{S4FPF}$  have the potential to be applied in a wide variety of settings.

%
%
%
 \bibliographystyle{plainurl}
 \bibliography{references/lit}

\clearpage

\clearpage

\begin{appendix}

\section{Proofs of the lower bounds}
\label{app:lower}

\lowerfin*

\begin{proof}
For any transition system $\cT=(T,\rightarrow,I , \AP, L)$, the relation $\abstracts$ on the class $\cF_{\cT}$ of finite abstractions  of $\cT$ is directed.
To see this consider finite abstractions $\cS_1=(S_1,\rightarrow_1,I_1 , \AP, L_1)$ and $\cS_2=(S_2,\rightarrow_2,I_2 , \AP, L_2)$ with abstraction functions $f_1\colon T \to S_1$ and $f_2\colon T\to S_2$.
We can define a common refinement $\cS$ on a subset of the finite  space $S_1\times S_2$ as follows: We define $f\colon T \to S_1\times S_2$ by $f(t) = (f_1(t),f_2(t))$ for all $t\in T$
and let $S = f(T)$ be the state space of the common refinement. As relation $\rightarrow'$ on $S$ we take the smallest relation such that $t\rightarrow t'$ implies
$f(t) \rightarrow' f(t')$ for all $t,t'\in T$. The  labeling of $f(t)$ is $L(t)$ for all $t\in T$. Finally a state  $s\in S$ is initial, if there is a $t\in I$ with $f(t) = s$. By definition, this results in an abstraction of $\cT$ witnessed by $f$.

To see that $\cS_1$ is an abstraction of $\cS$, consider the abstraction function $\pi_1\colon S \to S_1$ that simply projects states $(s_1,s_2)\in S$ onto the first component.
The condition on the  labeling is met as both $\cS$ and $\cS_1$ inherit the  labeling from $\cT$. If there is a transition $s\rightarrow_1 s'$ in $\cS_1$, then there are $t,t'\in T$ with $t\rightarrow t'$ and $f_1(t)=s$ and $f_1(t')=s'$.
But then, we have $f(t) \rightarrow' f(t')$ in $\cS$ and $\pi_1(f(t)) = s$ and  $\pi_1(f(t')) = s'$. So, the condition on the relation is met for the abstraction function $\pi_1$. The argument for initial states works analogously. Likewise, we can show that also $\cS_2$ is an abstraction of $\cS$ analogously. So, $\abstracts$ is directed on $\cF_\cT$.
\qed \end{proof}

\lowerall*

\begin{proof}
For the frame $(\cA_\cT, \abstracts)$ directedness is immediate as $\cT\in \cA_\cT$ is a common refinement of all abstractions in $\cA_\cT$. Furthermore, axiom $\mathsf{(.1)} = \lbox\ldia p \to \ldia\lbox p$ is valid on any transitive, reflexive frame $F$ with a greatest element $m$: For any world $w$ and any valuation $V(p)$ for $p$, $(F,V),w\Vdash \lbox\ldia p$ holds if and only if
$m\in V(p)$. But then,  also $(F,V),w\Vdash \ldia\lbox p$ as $(F,V),m\Vdash \lbox p$. So, the fact that $\cT\in \cA_\cT$ is a greatest element of the (general) frame, implies that $\mathsf{S4.2.1}\subseteq \MLARall_\cT$.
\qed \end{proof}

\proplower*

\begin{proof}
Consider an arbitrary transition system $\cT=(T,\rightarrow,I , \AP, L)$. Consider the refinement of $\cT$, which we obtain as follows:
For each  initial state $t\in I$, we pick an infinite path $\pi_t$ starting in $t$. 
Now, we consider a refinement $\cS$ that is the disjoint union of a copy of the original transition system $\cT$ where no state is initial anymore and transition systems that consist only of a copy of the path $\pi_t$ with the first state being initial for all $t\in I$. The abstraction function $f$ maps states from this disjoint union back to the states in $\cT$ of which they are a copy.
As states are mapped to their copies and $\cS$ contains a copy of $\cT$, the conditions on  labeling and the relation are met. Furthermore, each initial state $t\in T$ is the image under $f$ of the initial state of the copy of $\pi_t$.

Now, all paths in $\cS$ starting in an initial state are simply the disjoint copies $\pi_t$ for $t\in I$. A further refinement can further split up states, but it cannot introduce new traces or merge the paths. So, any CTL property that $\cS$ fulfills is also satisfied on all its refinements. 

Now, we show that $\mathsf{(.1)}= \lbox\ldia p \to \ldia\lbox p$ holds at $\cT$ in the general frame $G=G_{\mathfrak{A}, \abstracts, \mathrm{CTL}}$ where $\mathfrak{A}$ is the class of all transition systems.
Let   $V$ be a valuation on this general frame and suppose $(G,V),\cT \Vdash \lbox\ldia p $. So, $(G,V),\cS \Vdash \ldia p $. But as all refinements of $\cS$ satisfy the same CTL-properties, 
we have $(G,V),\cS \Vdash \lbox p $. So, $(G,V),\cT \Vdash \ldia \lbox p $ and hence $\mathsf{(.1)}$ is valid on $G$.
\qed \end{proof}

\section{Proof of Theorem \ref{thm:upper_fin}}
\label{app:upperfin}

\upperfin*

\begin{proof}
We will define a transition system $\cT$ and 
  show that for all $\cS \in \cF_\cT$,  we have $\MLARfin_{\cT}(\cS) \subseteq \mathsf{S4.2}$.
  Relying on \autoref{thm:S4.2}, we provide an infinite independent set of pure buttons and switches in $G_{\cF_\cT, \abstracts, \mathrm{CTL}}$ at any $\cS$.
The transition system   $\cT$ will be split into two  parts that are connected via a single shared initial state $s$ -- $\cT^\beta$ for the pure buttons and the $\cT^\sigma$ for the switches. Using different labels in $\cT^\beta$ and $\cT^\sigma$, we ensure that an abstraction of $\cT$ can be viewed as independently abstracting $\cT^\beta$ and $\cT^\sigma$.
  This also ensures that the pure buttons do not interfere with the switches and vice-versa.

\vspace{6pt}
\noindent
\textbf{Pure buttons.} 
We define $\cT^\button = (S, \rightarrow, \{s\}, \{s, f, a\}, L)$ sketched in Fig. \ref{fig:buttons_S4.2_app} where 
\begin{itemize}
\item
The state space is $S = \{s,f\} \cup \{(i,h,k)\in \mathbb{N}^3 \mid i\geq 2, i\geq h \geq 1\}$. 
\item
The relation $\rightarrow$ is given by 
\begin{itemize}
\item
$s\rightarrow (i,1,k)$ for all $i\geq 2$ and  all $k$,
\item
$(i,h,k)\rightarrow (i,h+1,k)$ for all $i\geq 2$, all $h<i$, and all $k$,
\item
$(i,i,k) \rightarrow f$ for all $i\geq 2$ and all $k$, and 
$f\rightarrow f$.
\end{itemize}
\item
The labeling over atomic propositions $\{s,f,a\}$ is given by $L(s)=\{s\}$, $L(f)=\{f\}$, and $L((i,h,k)) = \{a\}$ for all $i,h,k$.
\end{itemize}
\vspace{-6pt}

\noindent
Given a finite abstraction $\cS$ of $\cT^\beta$ with abstraction function $f$, we denote the set of states  mapped to the same state as  $(i,h,k)$ by $[(i,h,k)]$.
Overloading the notation, we  identify the state $f((i,h,k))$ with this equivalence class $[(i,h,k)]$.

The idea behind the $i$th pure button is roughly to say that some path along $i$ states labelled with $a$ from $s$ to $f$ has been ``isolated'' and not merged with paths of other length.
For any $i \ge 2$, we define the CTL-sentence
\[\button_i = \exists \lnext (\forall \lnext^{i-1} (a \land \forall \lnext f)).\]
In a finite abstraction $\cS$ of $\cT^\beta$, $\button_i$ holds iff there exists $k \geq 1$ such that $a^if^\omega$ is the only trace of paths starting in $[(i,1,k)]$.
 Once the only trace of paths starting in $[(i,1,k)]$ is $a^if^\omega$, this cannot be changed in a finer abstraction of $\cT^\beta$.

Further, whenever $\button_i$ is unpushed at an abstract transition system $\mathcal{S}_0$ then there is an accessible abstraction $\mathcal{S}_1$ of $\cT^\beta$ where $\button_i$ is pushed. In a finite abstraction there is always a path $[(i,1,k)],\ldots,[(i,i,k)]$ for some $k \geq 1$ that only visits infinite equivalence classes. Otherwise we would have infinitely many finite equivalence classes which leads to a contradiction. Thus, it is possible to safely (without pushing other pure buttons) split off the path $(i,1,k),\ldots,(i,i,k)$ by going to a refinement $\mathcal{S}_1$ that is even equivalence class preserving. 
Note that at any finite transition system $\cS$ only finitely many $\button_i$ can be pushed. So, for any $n \in \mathbb{N}$, there is a $I \subset \mathbb{N}$ with $|I| = n$ such that for any $i \in I$, $\button_i$ is unpushed.

		  \begin{figure*}[t]
		    \begin{subfigure}[b]{0.45\textwidth}
\centering
    \resizebox{1\textwidth}{!}{%
      \begin{tikzpicture}[scale=1,auto,node distance=4mm,>=latex]
        \tikzstyle{rect}=[draw=white,rectangle]

        \node[minimum size=6mm] (a210) {$(2,1,0)$};
         \node[minimum size=6mm,above= of a210] (a220) {$(2,2,0)$};
         
          \node[minimum size=6mm,right=3mm of a210] (a211) {$(2,1,1)$};
         \node[minimum size=6mm,above= of a211] (a221) {$(2,2,1)$};

         \node[minimum size=6mm,right= 0mm of a211] (a212) {$\dots$};

                                \draw[color=black ,  thick] (a210) ->  (a220);
			 \draw[color=black ,  thick] (a211) ->  (a221);

			    \node[minimum size=6mm,right= 3mm of a212] (a310) {$(3,1,0)$};
         \node[minimum size=6mm,above= of a310] (a320) {$(3,2,0)$};
         \node[minimum size=6mm,above= of a320] (a330) {$(3,3,0)$};
         
          \node[minimum size=6mm,right=3mm of a310] (a311) {$(3,1,1)$};
         \node[minimum size=6mm,above= of a311] (a321) {$(3,2,1)$};
         \node[minimum size=6mm,above= of a321] (a331) {$(3,3,1)$};

         \node[minimum size=6mm,right=0mm of a311] (a312) {$\dots$};
          \node[minimum size=6mm,right=0mm of a312] (a313) {$\dots$};
         

           \node[minimum size=6mm,below= of a310] (s) {$s$};
           
             \node[minimum size=6mm,above= of a330] (f) {$f$};
          
             \draw[color=black ,  thick] (a220) ->  (f);
              \draw[color=black ,  thick] (a221) ->  (f);
              
                  \draw[color=black ,  thick] (a330) ->  (f);
              \draw[color=black ,  thick] (a331) ->  (f);

              \draw[color=black ,  thick,->] (f) edge [loop above]   (f);

            \draw[color=black ,  thick] (s) ->  (a310);
              \draw[color=black ,  thick] (s) ->  (a311);

            \draw[color=black ,  thick] (s) ->  (a210);
              \draw[color=black ,  thick] (s) ->  (a211);

                                \draw[color=black ,  thick] (a310) ->  (a320);
                                  \draw[color=black ,  thick] (a320) ->  (a330);
			 \draw[color=black ,  thick] (a311) ->  (a321);
                                  \draw[color=black ,  thick] (a321) ->  (a331);

      \end{tikzpicture}
    }
\caption{Transition system  $\cT^\beta$.}
  \label{fig:buttons_S4.2_app}
  \end{subfigure}
      \begin{subfigure}[b]{0.45\textwidth}
\centering
    \resizebox{1\textwidth}{!}{%
      \begin{tikzpicture}[scale=1,auto,node distance=4mm,>=latex]
        \tikzstyle{rect}=[draw=white,rectangle]

        \node[minimum size=6mm] (a310k) {$(3,1,0,k)$};
         \node[minimum size=6mm,above= of a310k] (a320k) {$(3,2,0,k)$};
                  \node[minimum size=6mm,above= of a320k] (a330k) {$(3,3,0,k)$};

                   \node[minimum size=6mm,right=3mm of a310k] (a3x0k) {$(3,\ast,0,k)$};

                          \node[minimum size=6mm,right=3mm of a3x0k] (a311k) {$(3,1,1,k)$};
         \node[minimum size=6mm,above= of a311k] (a321k) {$(3,2,1,k)$};
                  \node[minimum size=6mm,above= of a321k] (a331k) {$(3,3,1,k)$};
                  
                    \node[minimum size=6mm,right=3mm of a311k] (a3x1k) {$(3,\ast,1,k)$};

         \node[minimum size=6mm,right= 6mm of a3x1k] (a4) {$\dots$};

                    \node[minimum size=6mm,below= 12mm of a311k] (s) {$s$};
            \node[minimum size=6mm,above=  of a331k] (f) {$f$};
            
                          \draw[color=black ,  thick,->] (f) edge [loop above]   (f);

             \draw[color=black ,  thick] (s) edge [bend left=15]  (a310k);
              \draw[color=black ,  thick] (s) edge [bend right=15]  (a4);
             \draw[color=black ,  thick] (s) ->  (a311k);
             \draw[color=black ,  thick] (a310k) ->  (a320k);
             \draw[color=black ,  thick] (a320k) ->  (a330k);
             \draw[color=black ,  thick] (a311k) ->  (a321k);
             \draw[color=black ,  thick] (a321k) ->  (a331k);
                          \draw[color=black ,  thick] (a331k) ->  (f);
             \draw[color=black ,  thick] (a330k) ->  (f);
                          \draw[color=black ,  thick] (a310k) ->  (a3x0k);
                                                    \draw[color=black ,  thick] (a311k) ->  (a3x1k);
                                                  \draw[color=black ,  thick] (a3x0k) ->  (a311k);
                                      \draw[color=black ,  thick] (a3x1k) ->  (a4);
                                      \draw[color=black ,  thick] (a3x1k) edge [bend left]  (a310k);
                                       \draw[color=black ,  thick] (a3x0k) edge [bend right]  (a4);

      \end{tikzpicture}
    }
\caption{Excerpt of   $\cT^\sigma$.}
  \label{fig:switches_S4.2_app}
  \end{subfigure}
  \caption{The components of the transition system $\cT$ in the proof of \Cref{thm:upper_fin}.}
\label{fig:S4.2_app}
\end{figure*}
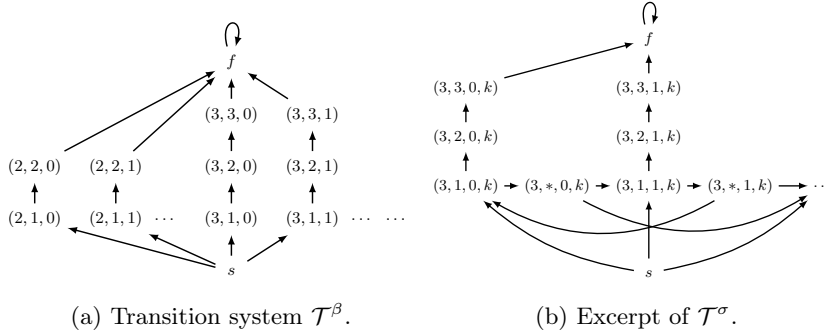

\vspace{6pt}
\noindent
\textbf{Switches.}
We define $\cT^\switch = (S, \rightarrow, \{s\}, \{s, f, b, x\}, L)$ where 
\begin{itemize}
\item
The state space is $S = \{s,f\}\cup \{(i,h,\ell, k) \mid i\geq 3,  h \in \{\ast,1,\dots, i\}, \ell\in \mathbb{N}, k\in \mathbb{N}\} $.
\item The relation $\rightarrow$ is given by 
\begin{itemize}
\item
$s\rightarrow (i,1,\ell,k)$ for all $i\geq 3$ and all $\ell,k$,
\item
$(i,h,\ell,k) \rightarrow (i,h,\ell,k)$ for all $i\geq 3$, $1\leq h <i$ and all $\ell,k$,
\item
$(i,1,\ell,k) \rightarrow (i,\ast,\ell,k)$ for all $i\geq 3$ and all $\ell,k$,
\item
$(i,\ast,\ell,k) \to (i,1,\ell',k)$ for all $i\geq 3$, all $\ell'\not= \ell$, and all $k$,
\item
$(i,i,\ell,k)\to f$ for all $i\geq 3$ and all $\ell, k$ as well as $f\to f$.
\end{itemize}
\item
The labeling over $\{s,f,x,b\}$ is given by $L(s)=\{s\}$, $L(f)=\{f\}$, $ L((i,h,\ell,k))= \{b\}$ for all $i\geq 3$, $1\leq h \leq i$, and all $\ell,k$, and
$ L((i,\ast,\ell,k))= \{x\}$ for all $i\geq 3$ and all $\ell,k$.
\end{itemize}
An excerpt of $\cT^\sigma$ for fixed length of the paths $i=3$ and fixed copy $k$ is sketched in Figure \ref{fig:switches_S4.2_app}.
From the state $(3,1,0,k)$, one can move via $(3,\ast,0,k)$ to all  states of the form $(3,1,\ell',k)$ with $\ell'\not= 0$. Further, a path of $b$-labelled states of length $3$ starts from 
$(3,1,0,k)$ towards $f$. The transition system $\cT^\sigma$ contains infinitely many copies of this structure indicated by the value $k\in \mathbb{N}$. Further, $\cT^\sigma$ contains this pattern
for all other length of paths $i>3$, which is denoted in the first component of the vectors.

The idea behind the $j$th switch is to say that one can find a copy of this pattern for path length $j$ where one $b$-path to $f$ of length $j$ has been ``isolated'' in the refinement
while all the other $b$-paths reachable via the $x$-labelled states with $\ast$ in the second component are not yet ``isolated''.
The switch can be turned on by isolating such a path. It can be turned off by isolating a second such $b$-path in the same copy. The infinitely many copies ensure that
there are always non-isolated fresh copies in any finite abstraction.

Formally,  for any $j \ge 3$, we define the  CTL-sentence
\[ \switch_j = \exists \lnext (\phi_j \land \forall \lnext (x \imp \forall \lnext \lnot \phi_j)) \quad \text{ where } \quad \phi_j =\forall \lnext (b \imp \forall \lnext^{j-2} (b \land \forall \lnext f)).\]
The switch $\switch_j$ holds at an abstraction $\cS$ iff there exists $k \in \mathbb{N}$ and there exists exactly one $ \ell \in \mathbb{N}$ such that $b^jf^\omega$ is the only trace from $[(j,1,\ell,k)]$ that does have label $b$ after one step.

Switching $\switch_j$ on and off can be done similarly to how we push our pure buttons. First, consider a transition system $\mathcal{S}_0$ where $\switch_j$ is off. Note that in a finite abstraction $\cS$ 
of $\cT^\sigma$ for almost all $k\in \mathbb{N}$ there is no $\ell \in \mathbb{N}$ such that $b^jf^\omega$ is the only trace from $[(j,1,\ell,k)]$ that has a $b$ in the second position. 
Fix one of these $k$. Further, as before, there is always a path $[(j,1,\ell,k)], \dots , [(j,j,\ell,k)]$ for some $\ell \in \mathbb{N}$ that only visits infinite equivalence classes. So, going from $\mathcal{S}_0$, we can safely (without switching other switches) split off the path $(j,1,\ell,k), \dots , (j,j,\ell,k)$  to obtain a transition system $\mathcal{S}_1$ by an equivalence class preserving refinement. Then, 
the switch $\switch_j$ is on at $\mathcal{S}_1$.

Secondly, consider a transition system $\mathcal{S}_0$ where $\switch_j$ is on. Note that there are possibly multiple, but finitely many states $[(j,1,\ell,k)]$ from which $b^jf^\omega$ is the only trace with $b$ in the second position. Choose one such state for some fixed $k,\ell \in \mathbb{N}$ as a candidate. Once again, there is always a path $[(j,1,\ell',k)], \dots , [(j,j,\ell',k)]$ for some $\ell' \neq \ell$ that only visits infinite equivalence classes. Thus, we can safely (without switching other switches) split off the path $(j,1,\ell',k), \dots , (j,j,\ell',k)$ and repeat this process for all other candidates to obtain $\mathcal{S}_1$ by an (equivalence class preserving) refinement. Then, $\switch_j$ is on at $\mathcal{S}_1$.

All in all, we have shown that for any $\cS$ in $\cF_\cT$ for the constructed transition system $\cT$, there are independent infinite families of unpushed pure buttons and switches.
By \Cref{thm:S4.2}, this implies that $\MLARfin_{\cT}(\cS) \subseteq \mathsf{S4.2}$ and hence finishes the proof of \Cref{thm:upper_fin}.
\qed \end{proof}

\section{Proof of Theorem \ref{thm:upperMLARall}}
\label{app:MLARall}

\upperboundMLARall*

\begin{proof}
In order to prove this theorem, we construct transition systems $\cT$ and $\cS \in \cA_\cT$ and show that $\MLARall_{\cT}(\cS) \subseteq \mathsf{S4.2.1}$. By \autoref{thm:S4.2.1}, it suffices to provide arbitrarily many independent unpushed pure buttons and $\Button$-restricted switches in $G_{\cA_\cT, \abstracts, \mathrm{CTL}}$ at $\cS$. Following the same idea as before, to guarantee independence, we construct $\cT$ and $\cS$ both split into different parts $\cT^\beta$ and $\cT^\switchB$ as well as $\cS^\beta$ and $\cS^\switchB$  for the unpushed pure buttons and for the $\Button$-restricted switches, respectively. The only state they share is a starting state $s$.

\begin{figure*}[t]
    \begin{subfigure}[b]{0.4\textwidth}
    \centering
    \resizebox{1\textwidth}{!}{%
    \begin{tikzpicture}[scale=1,auto,node distance=4mm,>=latex]
        \tikzstyle{rect}=[draw=white,rectangle]
        \node[initial above, label=125:$\{s\}$]         (s) at ( 0, 0) {$s$};
        \node[label=90:{$\{a_0\}$}]               (a11) at (-3, -1.5) {$a_0^1$};
        \node[label=180:{$\{a_1\}$}]               (a21) at (0, -1.5) {$a_1^1$};
        \node[]                                              (a31) at (3, -1.5) {$\cdots$};

        \node[label=270:{$\{a_0\}$}]               (a12) at (-3, -2.5) {$a_0^2$};
        \node[label=180:{$\{a_1\}$}]               (a22) at (0, -2.5) {$a_1^2$};
        \node[]                                              (a32) at (3, -2.5) {$\cdots$};
        \node[label=215:{$\{f\}$}]                      (f) at (0, -4) {$f$};
        \draw
        (s)     edge[below]                                node{} (a11)
        (s)     edge[below]                                node{} (a21)
        (s)     edge[dashed, below, shorten >=12pt]        node{} (a31)
        (a11)    edge[below]                               node{} (a12)
        (a21)    edge[below]                               node{} (a22)
        (a12)    edge[below]                               node{} (f)
        (a22)    edge[below]                               node{} (f)
        (a32)    edge[dashed, below, shorten <=12pt]       node{} (f)
        (f)     edge[in=290,out=340,loop]                  node{} (f)
        ;  
    \end{tikzpicture}
    }
    \caption{Transition system  $\cT^\beta$.}
    \label{fig:buttons_S4.2.1}
    \end{subfigure}
    \begin{subfigure}[b]{0.5\textwidth}
    \centering
    \resizebox{1\textwidth}{!}{%
    \begin{tikzpicture}[scale=1,auto,node distance=4mm,>=latex]
        \tikzstyle{rect}=[draw=white,rectangle]
        \node[initial above, label=125:$\{s\}$]         (s) at ( 0, 1) {$s$};
        \node[label=125:{$\{t\}$}]                    (b10) at (-4, -0.5) {$(j,\ast,0)$};
        \node[label=125:{$\{t\}$}]                    (b20) at (0, -0.5) {$(j,\ast,1)$};
        \node[]                                              (b30) at (4, -0.5) {$\cdots$};
        \node[label=180:{$\{b, b_j, b_j^0\}$}]        (b11) at (-4, -1.5) {$(j,1,0)$};
        \node[label=180:{$\{b, b_j, b_j^1\}$}]        (b21) at (0, -1.5) {$(j,1,1)$};
        \node[]                                              (b31) at (4, -1.5) {$\cdots$};

        \node[label=180:{$\{b, b_j, b_j^0\}$}]        (b12) at (-4, -2.5) {$(j,2,0)$};
        \node[label=180:{$\{b, b_j, b_j^1\}$}]        (b22) at (0, -2.5) {$(j,2,1)$};
        \node[]                                              (b32) at (4, -2.5) {$\cdots$};
        
        \node[label=180:{$\{b, b_j, b_j^0\}$}]        (b13) at (-4, -3.5) {$(j,3,0)$};
        \node[label=180:{$\{b, b_j, b_j^1\}$}]        (b23) at (0, -3.5) {$(j,3,1)$};
        \node[]                                              (b33) at (4, -3.5) {$\cdots$};
        \node[label=215:{$\{f\}$}]                      (f) at (0, -5) {$f$};
        \draw
        (s)     edge[below]                                node{} (b10)
        (s)     edge[below]                                node{} (b20)
        (s)     edge[dashed, below, shorten >=12pt]        node{} (b30)
        (b10)    edge[right]                               node{} (b20)
        (b20)    edge[dashed, right, shorten >=12pt]       node{} (b30)
        (b10)    edge[below]                               node{} (b11)
        (b20)    edge[below]                               node{} (b21)
        (b11)    edge[below]                               node{} (b12)
        (b21)    edge[below]                               node{} (b22)
        (b12)    edge[below]                               node{} (b13)
        (b22)    edge[below]                               node{} (b23)
        (b13)    edge[below]                               node{} (f)
        (b23)    edge[below]                               node{} (f)
        (b33)    edge[dashed, below, shorten <=12pt]       node{} (f)
        (f)     edge[in=290,out=340,loop]                  node{} (f)
        ;  
    \end{tikzpicture}
    }
    \caption{Excerpt of the transition system  $\cT^\switchB$.}
    \label{fig:switches_S4.2.1}
    \end{subfigure}
    \caption{The components of the transition system $\cT$ in the proof of \Cref{thm:upper_all}.}
    \label{fig:S4.2.1}
\end{figure*}
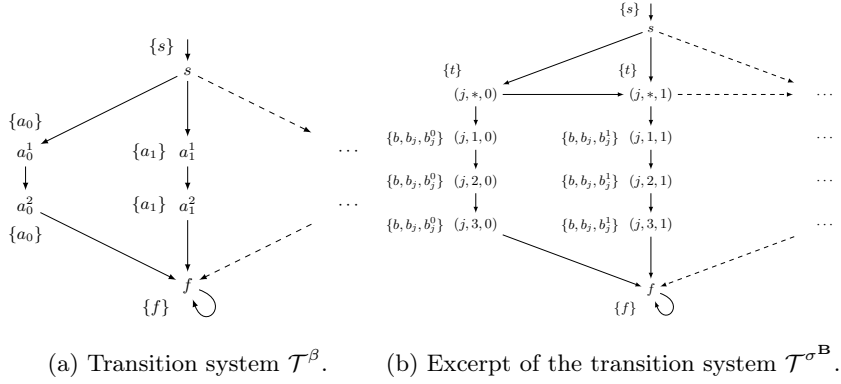

\vspace{6pt}
\noindent
\textbf{Unpushed pure buttons.}
We define $\cT^\button = (S, \rightarrow, \{s\}, \AP, L)$ sketched also in Figure \ref{fig:buttons_S4.2.1}:
\begin{itemize}
    \item
    The state space is $S = \{s, f\} \cup \{a_i^h \mid i \in \mathbb{N}, h \in \{1, 2\}\}$.
    \item The relation $\rightarrow$ is given by 
    $s \rightarrow a_i^1$,
    $a_i^1 \rightarrow a_i^2$, and
    $a_i^2 \rightarrow f$ for all $i$ as well as $f \to f$.
    \item
    The set of atomic propositions is $\AP = \{s, f\} \cup \{a_i \mid i \in \mathbb{N}\}$.
    \item
    The labeling over $\AP$ is given by $L(s) = \{s\}$, $L(f) = \{f\}$ and $L(a_i^h) = \{a_i\}$ for all $i,h$.
\end{itemize}
The abstraction $\cS^\button$ is given by collapsing the states $a_i^1$ and $a_i^2$ for all $i$. 
So, in  $\cS^\button$, we have   $[s] \rightarrow [a_i^1]=[a_i^2]$, $[a_i^1] \rightarrow [a_i^1]$, and  $[a_i^1] \rightarrow [f]$ for all $i$ as well as  $[f] \to [f]$.
   
 For any $i$, we define the CTL-sentence
\[\button_i = \exists \lnext (a_i \land \exists \lnext (a_i \land \forall \lnext f)).\]
This button holds at a transition system $\cS^\beta \abstracts \cR \abstracts \cT^\beta$ iff the states $a_i^1$ and $a_i^2$ are not collapsed to a single state in $\cR$.

Further, $\beta_i$ is a pure button: if it is false in some transition system $\cR_0$, then there exists a state $[a_i^1]$ with $a_i^1, a_i^2 \in [a_i^1]$ and by going to a refinement $\cR_1$ where the state is split into two, it becomes true and will remain true. Thus, it is also pure. Clearly, it is unpushed in $\cS$ and it does not interfere with any other pure buttons. 

\vspace{6pt}
\noindent
$\Button$\textbf{-restricted switches.}
We define $\cT^\switchB = (S, \rightarrow, \{s\}, \AP, L)$ where
\begin{itemize}
    \item
    The state space is $S = \{s, f\} \cup \{(j,h,k) \mid j \in \mathbb{N}, h \in \{\ast, 1, 2, 3\}, k \in \mathbb{N}\}$.
    \item The relation $\rightarrow$ is given by 
    \begin{itemize}
    \item
    $s \rightarrow (j,\ast,k)$ and  $(j,\ast,k) \rightarrow (j,\ast,k+1)$  for all $j,k$,
    \item
    $(j,\ast,k) \rightarrow (j,1,k)$,  $(j,1,k) \rightarrow (j,2,k)$, $(j,2,k) \rightarrow (j,3,k)$, and $(j,3,k) \rightarrow f$ for all $j,k$,
    \item $f \to f$.       
    \end{itemize}
    \item
    The set of atomic propositions is $\AP = \{s, f, t, b\} \cup \{b_j \mid j \ge 1\} \cup \{b_j^k \mid j \in \mathbb{N}, k \in \mathbb{N}\}$.
    \item
    The labeling over $\AP$ is given by $L(s) = \{s\}$, $L(f) = \{f\}$, $L((j,\ast,k)) = \{t\}$ for all $j,k$ and $L((j,h,k)) = \{b, b_j, b_j^k\}$ for all $h \in \{1, 2, 3\}$ and all $j,k$.
\end{itemize}
An excerpt of $\cT^\switchB$ for a fixed copy $k$ is sketched in Figure \ref{fig:switches_S4.2.1}.
The abstraction $\cS^\switchB$ is given by keeping states $s$, $f$, and $(j,\ast,k) $ for all $j,k$ as they are. The states $(j,1,k)$, $(j,2,k)$, and $(j,3,k)$, however, are mapped to a single state
$[(j,1,k)]$.
We define the CTL-sentence 
\[\Button = \exists \lnext (t \land \exists \lglobally (t \land \lnot \exists \lnext^2 f)).\]
This formula holds at $\cS^\switchB \abstracts \cR \abstracts \cT^\switchB$ iff there exists a $j$ such that for  all but finitely many $k$, the states $(j,1,k)$ and $(j,3,k)$
are not collapsed to a single state anymore. So, this is a pure button that is unpushed at $\cS^\switchB$.
 For any $j$, we now define the CTL-sentence
\[\switchB_j = \exists \lnext  \exists \lnext(
    b_j \land
    \lnot \exists \lnext f \land
    \lnot \exists \lnext (b_j \land \forall \lnext (b_j \land \forall \lnext f))).
\]
This holds at  $\cS^\switchB \abstracts \cR \abstracts \cT^\switchB$
iff there exists a $k$ such that $[(j,1,k)] = [(j,2,k)]$ or $[(j,2,k)] = [(j,3,k)]$, but not $[(j,1,k)] = [(j,3,k)]$. 

We claim that it is possible to switch $\switchB_j$ on and off as long as $\Button$ is false and without pushing it. Take a transition system $\cR_0$ where $\Button$ is false. If $\switchB_j$ is on at $\cR_0$, then we take for all $k$ the states $[(j,2,k)]$ with $(j,1,k) \in [(j,2,k)]$ xor $(j,3,k) \in [(j,2,k)]$ and split them into two to obtain a refinement $\cR_1$ where $\switchB_j$ is off. As we do not split up equivalence classes of three states, this will not push $\Button$. Now, consider the case that $\switchB_j$ is on at $\cR_0$. Since $\Button$ is false, there exists a $k$ such that $(j,1,k) \equiv (j,2,k) \equiv (j,3,k)$. Now, by splitting $(j,3,k)$ off, we reach a refinement $\cR_1$ where $\switchB_j$ is on. Since we only changed one state, this does not push $\Button$. Lastly, independence is clear as we use different labels for each $j$.
The whole argument is completely analogous for equivalence-class preserving and arbitrary refinement as the chosen transition systems encode the equivalence classes using atomic propositions.
\qed \end{proof}

\section{Proof of Theorem \ref{thm:upper}}
\label{app:upper}

\thmupper*

\begin{figure*}[t]
    \begin{subfigure}[b]{0.9\textwidth}
    \centering
    \resizebox{.9\textwidth}{!}{%
        \begin{tikzpicture}[scale=1,auto,node distance=4mm,>=latex]
            \tikzstyle{rect}=[draw=white,rectangle]
            \node[state, initial left, label=125:$\{s\}$]         (s) at ( 0, 0)  {$s$};
            \node[state, label=125:{$\{a_1\}$}]                   (a1) at (2.5, 1.5)  {$a_0$};
            \node[state, label=215:{$\{b_1\}$}]                   (b1) at (2.5, -1.5) {$b_0$};
            \node[state, label=90 :{$\{c\}$}]                   (c1) at (5, 0)      {$c_0$};
            \node[state, label=125:{$\{a_2\}$}]                   (a2) at (7.5, 1.5)  {$a_1$};
            \node[state, label=215:{$\{b_2\}$}]                   (b2) at (7.5, -1.5) {$b_1$};
            \node[]                                               (c2) at (10,0)      {$\cdots$};
            \node[state, label=125:{$\{a_n\}$}]                   (an) at (12.5, 1.5)  {$a_n$};
            \node[state, label=215:{$\{b_n\}$}]                   (bn) at (12.5, -1.5) {$b_n$};
            \node[state, label=90 :{$\{c\}$}]                   (cn) at (15,0)       {$c_n$};
            \draw
            (s)     edge[right]                                node{} (a1)
            (s)     edge[right]                                node{} (b1)
            (a1)    edge[right]                                node{} (c1)
            (b1)    edge[right]                                node{} (c1)
            (c1)    edge[right]                                node{} (a2)
            (c1)    edge[right]                                node{} (b2)
            (a2)    edge[dashed, right, shorten >=12pt]        node{} (c2)
            (b2)    edge[dashed, right, shorten >=12pt]        node{} (c2)
            (c2)    edge[dashed, right, shorten <=12pt]        node{} (an)
            (c2)    edge[dashed, right, shorten <=12pt]        node{} (bn)
            (an)    edge[right]                                node{} (cn)
            (bn)    edge[right]                                node{} (cn)
            ;  
        \end{tikzpicture}
    }
    \caption{The transition system $\cS_{n,m}^\button$.}
    \label{fig:buttonsMLAR}
    \end{subfigure}
    \begin{subfigure}[b]{0.6\textwidth}
    \centering
    \resizebox{.9\textwidth}{!}{%
        \begin{tikzpicture}[scale=1,auto,node distance=4mm,>=latex]
            \tikzstyle{rect}=[draw=white,rectangle]
            \node[state, label=90 :{$\{c\}$}]                   (cn) at (15,0)       {$c_n$};
            \node[state, label=125:{$\{d\}$}]                     (d1) at (17.5, 3)    {$d^0$};
            \node[state, label=125:{$\{d\}$}]                     (d2) at (17.5, 0)    {$d^1$};
            \node[]                                               (d3) at (17.5, -3)   {$\vdots$};
            \node[state, label=125:{$\{e_1\}$}]                   (e1) at (20, 5)    {$e_0^0$};
            \node[state, label=125:{$\{e_2\}$}]                   (e2) at (20, 3)      {$e_1^0$};
            \node[]                                               (e3) at (20, 1)    {$\vdots$};
            \node[]                                               (e4) at (20, 2)    {};
            \node[]                                               (e5) at (20, 0)      {};
            \node[]                                               (e6) at (20, -2)    {};
            \node[]                                               (ex) at (22,  3)    {};
            \node[]                                               (ey) at (22, -3)    {};
            \node[state, label=180:{$\{f\}$}]                      (f) at (24, 0) {$f$};
            \draw
            (cn)    edge[right]                                node{} (d1)
            (cn)    edge[right]                                node{} (d2)
            (cn)    edge[dashed, right, shorten >=12pt]        node{} (d3)
            (d1)    edge[below]                                node{} (d2)
            (d1)    edge[dashed, below, bend right, shorten >=12pt] node{} (d3)
            (d2)    edge[dashed, below, shorten >=12pt]        node{} (d3)
            (d1)    edge[right]                                node{} (e1)
            (d1)    edge[right]                                node{} (e2)
            (d1)    edge[dashed, right, shorten >=12pt]        node{} (e3)
            (d2)    edge[dashed, right, shorten >=48pt]        node{} (e4)
            (d2)    edge[dashed, right, shorten >=36pt]        node{} (e5)
            (d2)    edge[dashed, right, shorten >=48pt]        node{} (e6)
            (ex)    edge[dashed, right, shorten <=30pt]        node{} (f)
            (ey)    edge[dashed, right, shorten <=30pt]        node{} (f)
            (f)     edge[in=335,out=25,loop]                  node{} (f)
            ;  
        \end{tikzpicture}
    }
    \caption{The transition system $\cS_{n,m}^\switchB$.}
    \label{fig:switchesMLAR}
    \end{subfigure}
    \caption{The components of the transition system used in the proof of Theorem \ref{thm:upper}.}
\end{figure*}
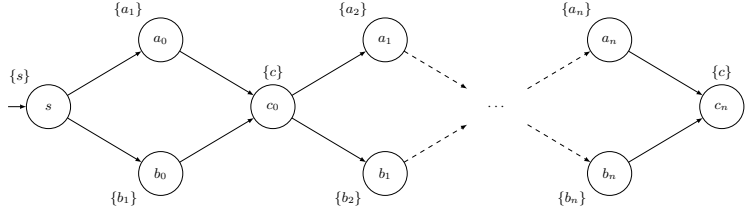
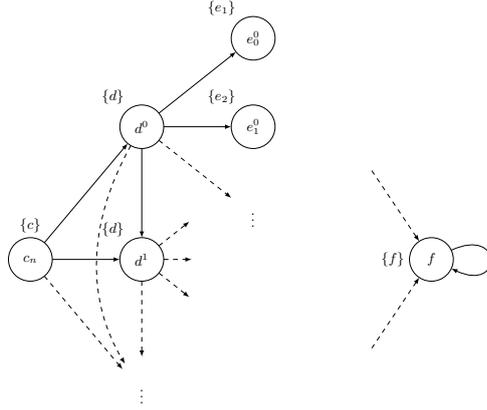

\begin{proof}
Using \autoref{cor:S4.2.1}, we construct a class of transition systems $\{\cS_{n,m}\}_{n,m \in \mathbb{N}}$ such that there is a set of unpushed pure buttons $\buttonSet = \{\button_i \mid 0 \le i \le n-1\} \cup \{\Button\}$ and a set of $\Button$-restricted switches $\switchSet = \{\switchB_j \mid 0 \le j \le m-1\}$ in $G_{\mathfrak{A}, \abstracts, \mathrm{CTL}}$ at $\cS_{n,m}$ where $\buttonSet \cup \switchSet$ is independent until $\Button$. To prevent inference, we again split the transition systems $\cS_{n,m} = \cS_{n,m}^\button \cup \cS_{n,m}^\switchB$ into different parts for the different types of control statements where the respective parts have a linking state $c_n$ in common. The general idea is to limit possible refinements so, we construct $\cS_{n,m}$ to not have cycles other than the self-loop at the finale state $f$. So, any refinement captured by CTL only makes states unreachable. 

\vspace{6pt}
\noindent
\textbf{Unpushed pure buttons.}
We define $\cS_{n,m}^\button = (S, \rightarrow, \{s\}, \AP, L)$ depicted also in Figure \ref{fig:buttonsMLAR} where
\begin{itemize}
    \item
    The state space is $S = \{s\} \cup \{a_i \mid i \le n-1\} \cup \{b_i \mid i \le n-1\} \cup \{c_i \mid i \le n-1\}$.
    \item The relation $\rightarrow$ is given by 
    \begin{itemize}
    \item
    $s \rightarrow a_1$ for all $i \le n-1$,
    \item
    $s \rightarrow b_1$ for all $i \le n-1$,
    \item
    $a_i \rightarrow c_i$ for all $i \le n-1$,
    \item
    $b_i \rightarrow c_i$ for all $i \le n-1$,
    \item
    $c_i \rightarrow a_{i+1}$ for all $i \le n-2$,
    \item
    $c_i \rightarrow b_{i+1}$ for all $i \le n-2$.
    \end{itemize}
    \item
    The set of atomic propositions is $\AP = \{s, c\} \cup \{a_i \mid i \le n-1\} \cup \{b_i \mid i \le n-1$\}.
    \item
    The  labeling over $\AP$ is given by $L(s) = \{s\}$, $L(a_i) = \{a_i\}$ for all $i \le n-1$, $L(b_i) = \{b_i\}$ for all $i \le n-1$ and $L(c_i) = \{c\}$ for all $i \le n-1$.
\end{itemize}

For any $i \le n-1$, we define the CTL-sentence
\[\button_i = \exists \lfinally (\forall \lnext a_i \lor \forall \lnext b_i),\]
which holds at a refinement $\cT$ of $\cS_{n,m}$ iff $a_i$ or $b_i$ is not reachable.
It is easy to see that $\button_i$ is a pure button and unpushed at $\cS_{n,m}$ since it is always possible to simply remove either $a_i$ or $b_i$ to push $\button_i$. Independence is clear as we use different labels.

\vspace{6pt}
\noindent
$\Button$\textbf{-restricted switches.}
We define $\cS_{n,m}^\switchB = (S, \rightarrow, \emptyset, \AP, L)$ depicted also in Figure \ref{fig:switchesMLAR} where
\begin{itemize}
    \item
    The state space is $S = \{c_n, f\} \cup \{d^k \mid k \in \mathbb{N}\} \cup \{e_j^k \mid j \le m-1, k \in \mathbb{N}\}$.
    \item The relation $\rightarrow$ is given by 
    \begin{itemize}
    \item
    $c_n \rightarrow d^k$ for all $k$,
    \item
    $d^k \rightarrow d^{k'}$ for all $k \le k'$,
    \item
    $d^k \rightarrow e_j^k$ for all $j \le m-1$ and all $k$,
    \item
    $e_j^k \rightarrow f$ for all $j \le m-1$ and all $k$, and $f \to f$.
    \end{itemize}
    \item
    The set of atomic propositions is $\AP = \{c, d, f\} \cup \{e_j : j \ge 1\}$.
    \item
    The  labeling over $\AP$ is given by $L(c_n) = \{c\}$, $L(d^k) = \{d\}$ for all $k$, $L(e_j^k) = \{e_j\}$ for all $j \le m-1$ and all $k$.
\end{itemize}

We define the CTL-sentence
\[B = \bigwedge_{j \le m-1} \exists \lfinally \exists \lglobally (d \land \lnot \exists \lnext e_j),\]
which holds in a refinement $\cT$ of $\cS_{n,m}$ iff there exists a $j \le m-1$ and there exists infinitely many $k$ such that $e_j^k$ is not reachable. So, $\Button$ is a unpushed pure button.

For any $j \le m-1$, we define the CTL-sentence
\[\switchB_j = \exists \lfinally (d \land \lnot \exists \lnext e_j),\]
which holds in a refinement $\cT$ of $\cS_{n,m}$ iff there exists a $k$ such that $e_j^k$ is not reachable.

This is a $\Button$-restricted switch: Consider a transition system $\cT_0$ where $\Button$ is false. Start with the case that $\switchB_j$ is off at $\cT_0$. Since $\Button$ is false at $\cT_0$, there exists a $k$ such that $e_j^k$ is reachable. Remove it and we have refinement $\cT_1$ where $\switchB_j$ is on. Since we only removed one state, this does not push $\Button$. Now, if $\switchB_j$ is on at $\cT_0$, things are a bit more tricky. Define a set $X \subseteq \mathbb{N}$ where for all $k$, $k \in X$ iff $e_j^k$ is not reachable. Then remove $d^k$ for all $k \in X$ and we get a refinement $\cT_1$ where $\switchB_j$ is off. Since $\Button$ is false at $\cT_0$, $X$ is only of finite size, so, removing all of these states will not push $\Button$. However, doing so, we might unintentionally switch other switches off, as well. Thus, we have to turn them on, again, using the procedure from above. So, in the end it is still possible to switch them independently.
\qed \end{proof}

\section{Proof of Theorem \ref{thm:S4FPF}}
\label{app:FPF}

\thmFPF*

\begin{proof}
Since $\mathsf{S4FPF}$ is characterized by the class of all finite partial function posets, if a formula $\phi \notin \mathsf{S4FPF}$, then there is a finite partial function poset $P = (F,\preccurlyeq)$ on $n$ elements, a valuation $V$ on $P$ and a initial state $f_0 \in F$ with $f_0(i) = \bot$ for all $0 \le i \le n-1$ such that $(F,V),w_0 \nSat \phi$.

We define formulas $\phi_{f} \in \mathcal{L}$ for every finite partial function $f \in F$
\[\phi_{f} = 
( \bigwedge\limits_{\substack{0 \le i \le n-1 \\ f(i) = 0}} \weakButton_i ) \land 
( \bigwedge\limits_{\substack{0 \le i \le n-1 \\ f(i) = 1}} \weakButtonAlt_i ).\]

It remains to show that $f \mapsto \phi_{f}$ is a $P$-labeling for $\GFw$ in $\GF$. For this, we check whether we meet the three labeling requirements:
\begin{enumerate}
    \item Clearly, every $\GFv \in \GFW$ satisfies exactly one of these formulas.
    \item Since the decisions are independent at $\GFw$, for every finite partial functions $g,h \in F$ and $\GFv \in \GFW$, if $\GFv \sat \phi_{g}$ then there is an $\GFu \in \GFW$ with $\GFv \GFR \GFu$ and $\GFu \sat \phi_{h}$ iff $g \preccurlyeq h$.
    \item Then, $\GFw \sat \phi_{f_0}$ since all decisions are unpushed at $\GFw$.
\end{enumerate}
By \autoref{lem:labeling}, since $(P,V),f_0 \nSat \phi$, we have $\phi \notin \Lambda_\GF(\GFw)$. So, for any formula $\phi \notin \mathsf{S4FPF}$, we have $\phi \notin \Lambda_\GF(\GFw)$ and thus $\Lambda_\GF(\GFw) \subseteq \mathsf{S4FPF}$.
\qed \end{proof}

The following formulation can be useful when it is not possible to find single states $c\in \cC$ with arbitrarily large independent sets of decisions.

\begin{corollary}
    Let $\GF$ be the general Kripke frame w.r.t. some $\GFW$, $\GFR$ and $\mathcal{L}$ as in \autoref{def:G}. If for any $n \in \mathbb{N}$, there is a state $\GFw \in \GFW$ and a set of unpushed decisions $\buttonSet = \{(\weakButton_i, \weakButtonAlt_i) \mid 0 \le i \le n-1\}$ in $\GF$ at $\GFw$ s.t. $\buttonSet$ is independent, then $\Lambda_\GF \subseteq \mathsf{S4FPF}$.
\end{corollary}

\begin{proof}
    By the same argument as in the previous proof, for any $n \in \mathbb{N}$, there exists a state $\GFw \in \GFW$ such that $f \mapsto \phi_{f}$ is a $P$-labeling for $\GFw$ in $\GF$. By \autoref{lem:labeling}, since $(P,V),f_0 \nSat \phi$, we have $\phi \notin \Lambda_\GF(\GFw)$ and thus, $\phi \notin \Lambda_\GF$. So, for any formula $\phi \notin \mathsf{S4FPF}$, we have $\phi \notin \Lambda_\GF$ and thus, $\Lambda_\GF \subseteq \mathsf{S4FPF}$.
\qed \end{proof}

\end{appendix}

\end{document}